\documentclass[12pt]{article}
\pdfoutput=1
\usepackage[margin=1in, centering]{geometry}
\usepackage{amsthm}
\usepackage[table]{xcolor}
\usepackage[normalem]{ulem}
\definecolor{TSUYUKUSA}{RGB}{46, 169, 223}
\definecolor{KURENAI}{RGB}{203, 27, 69}
\usepackage[pdfstartview=FitH,pdfpagemode=UseNone,colorlinks=true,citecolor=KURENAI,linkcolor=TSUYUKUSA,backref=page,linktocpage=true]{hyperref}

\newcommand {\mytitle} {Hypercontractivity for Quantum Erasure Channels via Variable Multipartite Log-Sobolev Inequality}

\usepackage{hyperref}
\usepackage{fullpage}
\usepackage{amssymb}
\usepackage{mathtools}
\usepackage{amsthm}
\usepackage{eucal}
\usepackage{dsfont}
\usepackage{libertine}
\usepackage[capitalise]{cleveref}
\usepackage{tikz}
\usepackage{thm-restate}
\usepackage{subfig}
\usepackage{environ}

\usetikzlibrary{arrows}

\makeatletter
\renewcommand{\section}{\@startsection {section}{1}{\z@}%
             {-3.5ex \@plus -1ex \@minus -.2ex}%
             {2.3ex \@plus .2ex}%
             {\normalfont\Large\scshape\bfseries}}
\renewcommand{\subsection}{\@startsection{subsection}{2}{\z@}%
             {-3.25ex\@plus -1ex \@minus -.2ex}%
             {1.5ex \@plus .2ex}%
             {\normalfont\large\scshape\bfseries}}
\renewcommand{\subsubsection}{\@startsection{subsubsection}{2}{\z@}%
             {-3.25ex\@plus -1ex \@minus -.2ex}%
             {1.5ex \@plus .2ex}%
             {\normalfont\normalsize\scshape\bfseries}}
\makeatother

\allowdisplaybreaks[1]

\def\X{\CMcal{X}}
\def\Y{\CMcal{Y}}
\def\Z{\CMcal{Z}}

\theoremstyle{plain}
\newtheorem{theorem}{Theorem}[section]
\newtheorem{lemma}[theorem]{Lemma}

\newtheorem{cor}[theorem]{Corollary}

\theoremstyle{definition}
\newtheorem{definition}[theorem]{Definition}
\newtheorem{claim}[theorem]{Claim}

\newtheorem{remark}[theorem]{Remark}
\newtheorem{fact}[theorem]{Fact}

\newcommand {\minusspace} {\: \! \!}

\newcommand {\Fn} [2] {\ensuremath{ #1 \minusspace \Br{ #2 } }}
\newcommand{\reals}{{\mathbb R}}
\newcommand {\set} [1] {\ensuremath{ \left\lbrace #1 \right\rbrace }}

\newcommand{\normthree}[1]{{\left\vert\kern-0.25ex\left\vert\kern-0.25ex\left\vert #1 \right\vert\kern-0.25ex\right\vert\kern-0.25ex\right\vert}}
\newcommand {\br} [1] {\ensuremath{ \left( #1 \right) }}
\newcommand {\Br} [1] {\ensuremath{ \left[ #1 \right] }}

\newcommand {\norm} [1] {\ensuremath{ \left\| #1 \right\| }}
\newcommand {\normsub} [2] {\ensuremath{ \norm{#1}_{#2} }}

\newcommand {\abs} [1] {\ensuremath{ \left| #1 \right| }}

\newcommand {\bra} [1] {\ensuremath{ \left\langle #1 \right| }}
\newcommand {\ket} [1] {\ensuremath{ \left| #1 \right\rangle }}
\newcommand {\ketbratwo} [2] {\ensuremath{ \left| #1 \middle\rangle \middle\langle #2 \right| }}
\newcommand {\ketbra} [1] {\ketbratwo{#1}{#1}}

\newcommand {\defeq} {\ensuremath{ = }}

\newcommand {\prob} [1] {\Fn{\Pr\,}{#1}}

\DeclareMathOperator*{\bigE}{\mathbb{E}}
\newcommand {\expec} [2] {\Fn{\bigE_{\substack{#1}}}{#2}}

\newcommand {\Tr} {\ensuremath{ \mathrm{Tr} }}

\newcommand {\id} {\ensuremath{\mathds{1}}}

\usetikzlibrary{calc}
\tikzset{meter/.append style={draw, inner sep=10, rectangle, font=\vphantom{A}, minimum width=30, line width=.8,
 path picture={\draw[black] ([shift={(.1,.3)}]path picture bounding box.south west) to[bend left=50] ([shift={(-.1,.3)}]path picture bounding box.south east);\draw[black,-latex] ([shift={(0,.1)}]path picture bounding box.south) -- ([shift={(.3,-.1)}]path picture bounding box.north);}}}

\newcommand {\Density} {\ensuremath{\mathcal{D}}} % symbol for (the set of) density operators
\newcommand {\PSD} {\ensuremath{\operatorname{Pos}}} % symbol for (the set of) PSD operators
\newcommand {\Herm} {\ensuremath{\mathcal{H}}} % symbol for (the set of) Hermitian operators
\newcommand {\Matrix} {\ensuremath{\mathcal{M}}} % symbol for (the set of) Matrices

\newcommand {\Errsym} {\ensuremath{\mathcal{E}}}
\newcommand {\Err} {\ensuremath{\ketbra{\Errsym}}} % this symbol: |E><E|
\newcommand {\Enteps} {\ensuremath{\mathrm{Ent}_\ep}} % Entropy operator with epsilon

\newcommand {\Ent} {\ensuremath{\mathrm{Ent}}}
\newcommand {\Hmin}[1] {\ensuremath{H_{\min}\left(#1\right)}}
\newcommand{\EPRN} {\ensuremath{\Phi^{\otimes n}}}

\newcommand {\wsum}[2] {\ensuremath{\sum\limits_{S\subseteq[#1]} (1-\ep)^{#1-\abs{S}}\ep^{\abs{S}} #2}}
\newcommand {\wsumr}[3] {\ensuremath{\sum\limits_{S\subseteq[#1]} (1-\ep)^{#2-\abs{S}}\ep^{\abs{S}} #3}}
\newcommand {\ps}[2] {\ensuremath{#1_{#2^c}}} % \ps{X}{S} X_S

\newcommand {\goodform} {\cref{app:def:goodform} form}

\newcommand {\Holder} {H\"{o}lder}

\usepackage{physics}
\renewcommand {\tr} {\ensuremath{\tau}} % normalized trace operator
\usepackage{bm}
\def\f{\frac}
\makeatletter
\DeclarePairedDelimiter{\@lrp}{(}{)}
\DeclarePairedDelimiter{\@lrb}{[}{]}
\def\p{\@lrp*}
\def\lrb{\@lrb*}
\def\lda{\lambda}
\def\bda{\bm{\lda}}
\def\ep{\varepsilon}
\def\diag{\textbf{diag}}
\makeatother
\def\co{J}
\def\ct{K}

\newcommand {\QEC} {\ensuremath{\mathcal{D}_\ep}} % Quantum Erasure Channel
\newcommand {\QECN} {\ensuremath{\QEC^{\otimes n}}} % n-fold product of the Quantum Erasure Channel
\newcommand {\QECR} {\ensuremath{\mathcal{D}}} % Noise Adder: D(A) = A + trA|EXE|
\newcommand {\QECRN} {\QECR^{\otimes n}} % n-fold product of Noise Adder

\newcommand {\Piep} {\ensuremath{\Pi_\ep}} % Pi_\ep
\newcommand {\PiepN} {\ensuremath{\Pi_\ep^{\otimes n}}} % n-fold product of Pi_\ep

\NewEnviron{Anonymous}{\ifx\AnonymousSwitch\undefined\BODY\fi}

\allowdisplaybreaks

\begin{document}

%\begin{titlepage}
\title{\mytitle\\[2ex]}

%Comment out this line to switch off anonymous
%\newcommand{\AnonymousSwitch}

%\author{
%	\authorblock{\Penghui}{\CQT, \NUS}{phyao1985@gmail.com}
%}
\begin{Anonymous}
\author{
    Zongbo Bao\thanks{\scriptsize State Key Laboratory for Novel Software Technology, New Cornerstone Science Laboratory, Nanjing University, Nanjing, China. Email: baozb0407@gmail.com.}
    \and Yangjing Dong\thanks{\scriptsize State Key Laboratory for Novel Software Technology, New Cornerstone Science Laboratory, Nanjing University, Nanjing, China. Email: dongmassimo@gmail.com.}
    \and Fengning Ou\thanks{\scriptsize State Key Laboratory for Novel Software Technology, New Cornerstone Science Laboratory, Nanjing University, Nanjing, China. Email: reverymoon@gmail.com.}
    \and Penghui Yao\thanks{\scriptsize State Key Laboratory for Novel Software Technology, New Cornerstone Science Laboratory, Nanjing University, Nanjing, China. Email: phyao1985@gmail.com.}~\thanks{\scriptsize Hefei National Laboratory, Hefei 230088, China.}
}
\end{Anonymous}

\clearpage\maketitle

\begin{abstract}
We prove an almost optimal hypercontractive inequality for products of quantum erasure channels, generalizing the hypercontractivity for classical binary erasure channels. To our knowledge, this is the first tensorization-type hypercontractivity bound for quantum channels with no fixed states. The traditional inductive arguments for classical hypercontractivity cannot be generalized to the quantum setting due to the nature of the non-commutativity of matrices. To overcome the difficulty, we establish a novel quantum log-Sobolev inequality for Bernoulli entropy, which includes the classical log-Sobolev inequality and the quantum log-Sobolev inequality as one-partite cases. To our knowledge, its classical counterpart is also unknown prior to this work. We establish a connection between our quantum log-Sobolev inequality and the hypercontractivity bound for quantum erasure channels via a refined quantum Gross' lemma, extending the analogous connection between the quantum log-Sobolev inequality and the hypercontractivity for qubit unital channels.  As an application, we prove an almost tight bound (up to a constant factor) on the classical communication complexity of two-party common randomness generation assisted with erased-noisy EPR states, generalizing the tight bound on the same task assisted with erased-noisy random strings due to Guruswami and Radhakrishnan.
\end{abstract}

%\end{titlepage}

\section{Introduction}\label{section:Introduction}
The notion of \emph{hypercontractivity} originates from quantum field theory~\cite{nelson1966quartic,cmp/1103840511,nelson1973free,SIMON1972121} and functional analysis~\cite{paley1932remarkable,Bonami1968,kiener1969uber,bonami1970etude}, which has a rich history of research(see~\cite[Notes in Chapter 9]{O14ABF} for more background and history). Hypercontractive inequalities nowadays are also a fundamental tool in the analysis of Boolean functions, which found wide applications in various areas of theoretical computer science~\cite{ KKL,10.1145/174130.174138,mossel2010noise,doi:10.1137/S0097539705447372,mossel2010noise}.

A typical hypercontractive inequality is the Bonami-Beckner inequality~\cite{Bonami1968,a329bd8a-47e2-3ad4-877f-4f9e7e7cccfa}. For a fixed $\rho\in[0,1]$, consider the noise operator on the space of all functions $f:\set{0,1}^n\rightarrow\mathbb{R}$, defined by
\[\br{T_{\rho}(f)}(x)=\expec{y}{f(y)},\]
where the expectation is taken over $y$ obtained from $x$ by flipping each bit independently with probability $(1-\rho)/2$. Intuitively, $T_{\rho}$ ``smooths'' the function $f$. That is, the peaks of $f$ are smoothed out in $T_{\rho}(f)$. Consider the $p$-norm of $f$: $\normsub{f}{p}=\br{\expec{x}{|f(x)|^p}}^{1/p}$, where the expectation is over the uniform distribution on $\set{0,1}^n$. The $p$-norm is monotone non-decreasing with respect to $p$. The hypercontractive inequality captures the smoothing effect of $T_{\rho}$ in terms of $p$-norms. More precisely, we have
\begin{equation}\label{eqn:HC}
  \norm{T_{\rho}(f)}_q\leq\norm{f}_p,
\end{equation}
as long as $1\leq p\leq q$ and $\rho\leq\sqrt{(p-1)/(q-1)}$. It has been further generalized to the non-binary and non-uniform settings~\cite{PawelWolff2007}.

Hypercontractive inequalities in non-commutative settings have received increasing attention since the quantum computing and quantum information technologies emerge \cite{10.1215/S0012-7094-75-04237-4,cmp/1104253198,OLKIEWICZ1999246,4690981,kastoryano2013quantum,king2014hypercontractivity,6a361ba5-4637-3342-86c3-f110678a9132,10.1063/5.0056388}. To put the classical hypercontractivity in the context of quantum computing, let's reformulate a Boolean function $f:\set{0,1}^n\rightarrow\mathbb{R}$ as a $2^n\times 2^n$ diagonal matrix with diagonal entries $f(x)$, i.e., $M_f=\sum_{x\in\set{0,1}^n}f(x)\ketbra{x}$. Then $T_{\rho}$ is equivalent to a product of quantum depolarizing channels acting on
$M_f$, and $\norm{f}_p$ is the normalized Schatten $p$-norm of $M_f$, which will be defined later. Thus it is natural to consider the hypercontractivity bound for the product of depolarizing channels acting on an arbitrary operator.
The main challenge in the non-commutative case is the tensorization property, i.e.,
to generalize the hypercontractive inequality for a product of $n$ quantum channels,
from the hypercontractive inequality for a local single copy quantum channel,
which is easy in the classical commutative case.
Kastoryano and Temme~\cite{kastoryano2013quantum} proved a hypercontractive inequality for the product of depolarizing channels. Furthermore, King~\cite{king2014hypercontractivity} generalized this result to the products of all {\em unital} qubit channels, which map identity matrices to identity matrices. Recently, Beigi~\cite{10.1063/5.0056388} further improved the hypercontractivity bound for the product of depolarizing channels with better parameters. Researchers have also established several variants of non-commutative hypercontractive inequalities, which have found applications in quantum communication complexity~\cite{10.1145/1250790.1250866},
quantum coding~\cite{4690981,arunachalam2023matrix}, quantum Markov semigroups~\cite{Bardet2022-bg} and quantum non-local games~\cite{doi:10.1137/20M134592X,qin_et_al:LIPIcs.ICALP.2023.97}, etc.

%Nair and Wang \cite{7541363} have proved a classical hypercontractivity bound for the \emph{binary erasure channel} (BEC), recently.
%For an erasure probability $\ep\in[0,1]$, the BEC erases the input $x$ and outputs ``$*$'' indicating error with probability $\ep$.
%A BEC is a non-unital operator, which is different from all the operators described above.
%Guruswami and Radhakrishnan \cite{guruswami_et_al:LIPIcs:2016:5845}
%used Nair and Wang's result to obtain a tight bound on the communication complexity of common randomness generation,
%in a setting where the communicating parties share classical correlation affected by binary erasure noise.
%Bogdanov and Prakriya \cite{bogdanov_et_al:LIPIcs.ICALP.2021.33} used the BEC hypercontractivity to
%present a new $4$-query direct sum test with an optimal soundness error.
%The reverse hypercontractivity for the BEC has also been studied in~\cite{8006666}. Very recently, Eldan, Wigderson and Wu~\cite{10.1145/3564246.3585205} have independently proved a hypercontractive inequality for binary erasure channels (in a different language), which was used to give a sharp version of the "it ain't over till it's over" theorem~\cite{10.1109/SFCS.2005.53}.

Besides the hypercontractivity for product of unital channels studied in~ \cite{king2014hypercontractivity,10.1063/5.0056388},
there are some results concerning the hypercontractivity on 
 non-unital channels in a different point of view 
 \cite{OLKIEWICZ1999246, Temme_2014,MHSFW16nonunital}.
They focused on the hypercontractive inequalities for an evolution under $L_p$ 
 norm with respect to the fixed state of the evolution, a.k.a. the stationary state of the evolution.
%It is easy to see the QEC does not have any fixed state and therefore these 
%results are not suitable for the QEC.

In this paper, we focus on the \emph{quantum erasure channel} (QEC). A QEC acts on the input qubit and outputs an error with probability $\ep$, and leaves the input state unchanged with probability $1-\ep$.
Formally, given an input qubit $\rho$, we define the quantum erasure channel $\QEC$ as
\begin{equation}\label{eq:QEC}
  \QEC(\rho)=(1-\ep)\cdot\rho+\ep\cdot\Err.
  %\QEC(\rho)\defeq(1-\ep)\cdot\ketbra{0}\otimes\rho+\ep\cdot\ketbra{1}\otimes\frac{\id_2}{2}.\footnote{The form is slightly different from the standard one\cite{watrous2018theory} \[\QEC(\rho)=(1-\ep)\rho+\ep\ketbra{\text{err}}.\]  It is not hard to see they are equivalent.}
\end{equation}
In case of the error (with probability $\ep$),
the input state $\rho$ is replaced by the state $\Err$ indicating an erasure error.
See~\cref{fig:ab:bec+qec} for an illustration.
The QEC is one of the most fundamental quantum channels, which have received extensive studies. 
There have been various quantum error correcting codes proposed for the quantum erasure channel \cite{grassl1997codes,DZ13quantumerasure,7541599,DZ20quantumerasure,ZOJ23quantumerasure}.
Besides, they are one of the few quantum channels whose asymptotic capacity for faithful transmission can be computed exactly \cite{bennett1997capacities}.

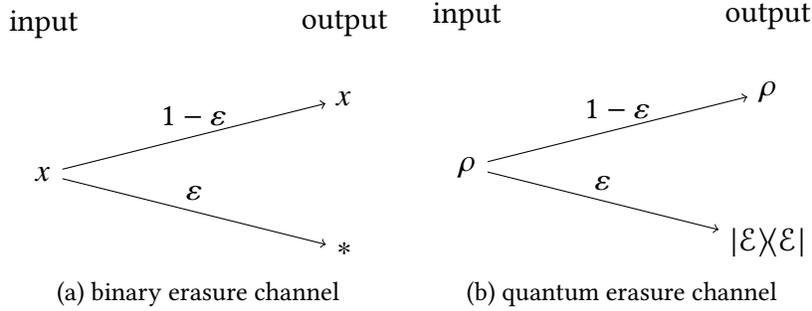
\begin{figure}
\centering
\subfloat[binary erasure channel]{
  \begin{tikzpicture}
    \node[] (li) at (0, 1) {input};
    \node[] (lo) at (4, 1) {output};
    \node[] (ix) at (0, -1) {$x$};
    \node[] (ox) at (4, 0) {$x$};
    \node[] (oe) at (4, -2) {$*$};
    \draw[->] (ix) -- node[above] {$1-\ep$} (ox);
    \draw[->] (ix) -- node[above] {$\ep$} (oe);
  \end{tikzpicture}
}\label{fig:a:bec}
\subfloat[quantum erasure channel]{
  \begin{tikzpicture}
    \node[] (li) at (0, 1) {input};
    \node[] (lo) at (4, 1) {output};
    \node[] (ix) at (0, -1) {$\rho$};
    \node[] (ox) at (4, 0) {$\rho$};
    \node[] (oe) at (4, -2) {$\Err$};
    \draw[->] (ix) -- node[above] {$1-\ep$} (ox);
    \draw[->] (ix) -- node[above] {$\ep$} (oe);
  \end{tikzpicture}
}\label{fig:b:qec}
\caption{Classical and quantum erasure channel. With probability $\ep$, error occurs and the input bit/state is erased.
  In the quantum erasure channel, the state $\Err$ indicating error replaces $\rho$ in case of erasure.}\label{fig:ab:bec+qec}
\end{figure}

In this paper, we establish the first hypercontractivity bound for the tensor product of QECs. 
It is interesting to note that an almost optimal hypercontractivity bound for quantum depolarizing channels
can be obtained from the hypercontractivity for QECs established in this paper\footnote{See \cref{appendix-proof-to-depolarizing-channel} for a proof}.
The hypercontractivity for classical binary erasure channels has been proved by  
Nair and Wang \cite{7541363} recently, which have found applications in common randomness generation \cite{guruswami_et_al:LIPIcs:2016:5845}  and property testing \cite{bogdanov_et_al:LIPIcs.ICALP.2021.33}.
Recently, Eldan, Wigderson and Wu~\cite{10.1145/3564246.3585205} have independently proved a hypercontractive inequality for the binary erasure channel (in a different language), which was used to give a sharp version of the ``it ain't over till it's over'' theorem~\cite{mossel2010noise}.

QECs  are non-unital and do not have fixed points as the input space and output space have different dimensions. Thus they are not covered by existing conclusions from 
\cite{king2014hypercontractivity,10.1063/5.0056388,OLKIEWICZ1999246, Temme_2014,MHSFW16nonunital}. To our knowledge, this is the first quantum channels with no fixed states, whose tensorized hypercontractivity is established.

%Since the QEC is non-unital, it is not covered by existed conclusions from 
% \cite{king2014hypercontractivity} and \cite{10.1063/5.0056388}.
%Before our work, there are some results concerning the hypercontractivity on 
% non-unital channels in a different point of view 
% \cite{OLKIEWICZ1999246, Temme_2014,MHSFW16nonunital}.
%They focused on the hypercontractive inequalities for an evolution under $L_p$ 
% norm with respect to the fixed state of the evolution.
%It is easy to see the QEC does not have any fixed state and therefore these 
% results are not suitable for the QEC.

The most well-known tensorization proofs for classical hypercontractivities are via the inductive argument on the length of the input.
For quantum erasure channels, by combining the inductive argument and the norm compression inequality, we could solve the case $1\le p\le 2\le q$.
However, due to the nature of the norm compression for the case $1\le p\le q\le 2$, this inductive argument fails to reach any reasonable conclusion.
To remedy it, the authors in \cite{OLKIEWICZ1999246,king2014hypercontractivity,10.1063/5.0056388} use a new framework for quantum hypercontractivity, which turns out to be very successful for several quantum channels.  A crucial component in this framework is the equivalence between the {\em log-Sobolev inequality} and hypercontractive inequalities via the {\em Gross' lemma}~\cite{10.1215/S0012-7094-75-04237-4,kastoryano2013quantum}.

Log-Sobolev inequalities are the functional inequalities that upper bound the entropy of a function by their Dirichlet forms.
Classical log-Sobolev inequalities are amongst the most studied functional inequalities for semigroups~\cite{BH99general, FMP13general, IM2014general,doi:10.1142/S0219025715500113,BB21general,FOW2022general}, which turn out to be a versatile tool to analyze mixing time of various Markov processes~\cite{bakry2014analysis,CGM19general, CCN21general}.
Olkiewicz and Zegarlinski~\cite{OLKIEWICZ1999246} introduced the log-Sobolev inequality in a general non-commutative setting,
which is further studied as \emph{quantum log-Sobolev inequalities} on finite dimensional state spaces~\cite{kastoryano2013quantum,Temme_2014}.
In quantum log-Sobolev inequalities, the entropy is defined on all quantum operators,
using a generalized definition of the entropy on functions in the commutative case.
This is in contrast with the standard quantum information theory definition of entropy defined primary for quantum states.
Quantum log-Sobolev inequalities have played a crucial role in the proof of hypercontractivity bounds for unital qubit channels~\cite{king2014hypercontractivity}
and have also found applications in the study of quantum systems~\cite{Carlen2020,PhysRevLett.130.060401}, noisy quantum devices~\cite{StilckFranca2021}, quantum Gibbs samplers~\cite{Kastoryano2016}, quantum Markov semigroups~\cite{Bardet2022-bg}, etc.

To our knowledge, all previous classical or quantum log-Sobolev inequalities are concerned with the systems with a fixed size (i.e., number of bits in the classical system or the number of qubits in the quantum system). However, the size of the system varies when considering BECs because of the extra error dimension  $\ket{\mathcal{E}}$ in the output system. Hence, we establish a {\em quantum log-Sobolev inequality for Bernoulli entropy}, where the size of the system varies subject to erasure noise. It includes Kastoryano and Temme's quantum log-Sobolev inequality~\cite{kastoryano2013quantum} as a single partite case.
Our quantum log-Sobolev inequality for Bernoulli entropy
 is essential to the proof of our QEC hypercontractivity, and is of independent interest as well. To our knowledge, even the classical analog of our log-Sobolev inequality is unknown. Additionally, the quantum Gross' lemma~\cite{10.1215/S0012-7094-75-04237-4,king2014hypercontractivity}, which bridges the hypercontractivity and log-Sobolev inequalities, is replaced by a refined inequality connecting the Schatten $p$ norm of an operator and the Schatten $p$ norm of its reduced operator. This lemma is used to induce a hypercontractivity bound for the QEC from the quantum log-Sobolev inequality for Bernoulli entropy.

As an application of our hypercontractive inequality for the QEC, we prove an almost tight bound (up to a constant factor) on the classical communication complexity of the two-party common randomness generation problem,
where maximally entangled states affected by the QEC are shared among the communicating parties. Common randomness generation, a classic topic in information theory and distributed computing, was first raised by Maurer~\cite{maurer1993secret} and Ahlswede and Csisz\'ar~\cite{ahlswede1998common}, leading to numerous subsequent works.
Devetak and Winter \cite{devetak2004distilling} introduced the problem of entanglement-assisted common randomness generation,
and studied the case of classical-quantum correlations.
Recently, Lami, Regula, Wang and Wilde~\cite{10161613} gave efficiently computable upper bounds on the LOCC assisted distillable randomness of a bipartite quantum state.
Also, an upper bound in terms of the fidelity-based smooth min-relative entropy was given by Nuradha and Wilde~\cite{nuradha2023fidelitybased}.
The idea of these two results are strongly linked to those in~\cite{10005080}. 
Readers may refer to the excellent survey~\cite{8863950} and the references therein.

Most of these works assume that the players share only a limited amount of noisy correlation.
In contrast, in this work we assume the noisy correlations are free and unlimited,
and try to minimize the classical communication.
Guruswami and Radhakrishnan~\cite{guruswami_et_al:LIPIcs:2016:5845} investigated the communication complexity of common randomness generation when the parties share unlimited but non-perfect correlation. They considered both the binary symmetric channel and the binary erasure channel and established tight tradeoffs among the communication complexity, min-entropy of the common randomness and errors for both cases.
Dong and Yao~\cite{DY2023} investigated common randomness generation when the players share unlimited depolarized noisy EPR states, and they also established a tight tradeoff among the classical communication complexity, min-entropy of the common randomness and errors. They also establish a tradeoff when quantum communication is allowed, which is not known to be tight.
However, the communication complexity of common randomness generation when the players sharing maximally entangled states effected by the quantum erasure channel remains unknown.
For the quantum erasure channel, this is due to a lack of a tensorized hypercontractivity bound,
which is a key component of all these previous proofs on the communication complexity of common randomness generation.
In this work, by using our hypercontractivity bound for the product of quantum erasure channel,
we are able to prove a communication lower bound when the players share maximally entangled states affected by the quantum erasure channel. This lower bound matches the classical communication upper bound of ~\cite{guruswami_et_al:LIPIcs:2016:5845}, where the players share classical correlation with the binary erasure channel, and thus is tight up to constant a factor.

\subsection{Our Results}
\paragraph{Hypercontractivity}
Our main result is a hypercontractivity bound for the product of quantum erasure channels.
Before illustrating our main result, we first state the hypercontractivity bound for the classical binary erasure channel proved by Nair and Wang~\cite{7541363}.
A similar inequality has been independently proved by Eldan, Wigderson and Wu~\cite{10.1145/3564246.3585205}.
\begin{theorem}[\cite{7541363}]
  Let $f: \set{-1,1}^{n}\to\mathbb{R}$ be a function,
  $X$ be a random variable uniformly distributed over $\set{-1,1}^{n}$
  and $Y=\operatorname{BEC}_\ep(X)$.
  Define the function $g: \set{-1,1,*}^n\to\mathbb{R}$ as
  \begin{equation}
    g(y) = \expec{}{f(X)\mid Y=y}.
  \end{equation}
  Then for any $\ep\in[0,1], p,q>1$ satisfying $1-\ep\le(p-1)/(q-1)$, it holds that
  \begin{equation}\label{intro:bec:hyper}
    \normsub{g(Y)}{q}\le\normsub{f(X)}{p}
  \end{equation}
  where the $\lambda$-norm is defined as
  \begin{equation}
    \normsub{Z}{\lambda}\defeq\expec{}{\abs{Z}^\lambda}^{\frac{1}{\lambda}}.
  \end{equation}
\end{theorem}
To express our hypercontractivity bounds, we need to decompose a QEC into two operations:
\begin{enumerate}
  \item Given a $2\times 2$ matrix $\rho$, we first define an expanding operator $\QECR$ that creates both the ``non-erased'' state and the ``erased'' state, that is
    \begin{equation}
      \QECR(\rho)\defeq \rho+\tr(\rho)\cdot\Err = \begin{pmatrix} \rho & \\&\tr(\rho)\end{pmatrix},
    \end{equation}
  where $\tr(\cdot)$ stands for the normalized trace operation,
    and here $\tr(\rho) = \f{1}{2}\Tr \rho$.
    This is the analog of the function $g(y) = \expec{}{f(X)|Y=y}$.
  \item Then we define the noise matrix $\Piep$ as
    \begin{equation}
    %\Pi_\ep\defeq\begin{bmatrix}1-\ep&\\&\ep\\\end{bmatrix}\otimes\id_2
    \Pi_\ep\defeq(1-\ep)\cdot\id_2+ 2\ep\cdot\Err=\begin{pmatrix}1-\ep &&\\&1-\ep&\\&&2\ep\end{pmatrix}.
    \end{equation}
    This matrix reassigns the noise, which simulates taking expectation over $Y$ in the classical case.
\end{enumerate}
Recall that in \cref{eq:QEC} we define
$$
    \QEC(\rho)=(1-\ep)\cdot\rho+\ep\cdot\Err.
$$
We observe that for any $2\times 2$ matrix $\rho$,
\begin{equation}
\QEC(\rho) = \Piep\cdot\QECR(\rho).
\end{equation}
We now state our hypercontractivity bound for the quantum erasure channel:
\begin{theorem}[informal version of \cref{hyper:thm:hc}]\label{introduction-hc-informal}
Given any $2^n\times 2^n$ matrix $X$, $1\ge\ep\ge0$, $q\ge p\ge 1$ and $c\ge 1$ satisfying $1-\ep\le\left(\frac{p-1}{q-1}\right)^c$, let $\normthree{\cdot}_{p}$ denote normalized Schatten $p$-norms.
The hypercontractive inequality for the quantum erasure channel
\begin{equation}\label{equation-qec-hc-introduction-d}
\left(2^{-n} \Tr\left[\PiepN\cdot\abs{\QECR^{\otimes n}(X)}^q\right]\right)^{1/q}\le\normthree{X}_p
\end{equation}
holds if one of the following conditions is satisfied:
\begin{itemize}
  \item $c=1$ and $1\le p\le 2\le q$.
  \item $c=2$ and $1\le p\le q\le 2$.
\end{itemize}
\end{theorem}
%\begin{cor}[Equivalent form]
%Given the setting of \Cref{introduction-hc-informal}, the following inequality
%  \begin{equation}
%    \br{\wsum{n}{\tr\Br{\br{\tr_SX}^q}}}^{1/q} \le \normthree{X}_p
%  \end{equation}
%holds if one of the following conditions is satisfied:
%\begin{itemize}
%  \item $c=1$ and $1\le p\le 2\le q$.
%  \item $c=2$ and $1\le p\le q\le 2$.
%\end{itemize}
%\end{cor}
\begin{remark}[Equivalent form, hypercontractivity for random partial trace] It is not hard to verify that Eq.~\eqref{equation-qec-hc-introduction-d} is equivalent to the following.
\begin{equation}
    \br{\wsum{n}{\normthree{\tr_SX}_q^q}}^{1/q} \le \normthree{X}_p,
  \end{equation}
where $\tau_SX=\Tr_SX/2^{|S|}$ is the normalized partial trace. Thus, an alternate interpretation of \Cref{introduction-hc-informal} is a hypercontractive inequality for random partial trace operations.
\end{remark}
\begin{remark}\label{rk:cq}
It might seem opaque that \cref{equation-qec-hc-introduction-d} is a direct quantum generalization of \cref{intro:bec:hyper}.
The connection becomes clearer when we notice that the expectation on the left-hand side of \cref{intro:bec:hyper} is taken over the random variable $Y=\operatorname{BEC}_\ep(X)$.
The random variable $Y$ is certainly not uniformly distributed, and we can verify that for each $i\in[n]$,
\begin{equation}\label{intro:bec:expec_distribution}
  Y_i=
  \begin{cases}
    0 &\text{ w.p. } (1-\ep)/2\\
    1 &\text{ w.p. } (1-\ep)/2\\
    * &\text{ w.p. } \ep\\
  \end{cases}
\end{equation}
We can expand the left-hand side of \cref{intro:bec:hyper} according to \cref{intro:bec:expec_distribution} and get
\begin{equation}\label{intro:bec:left_expanded}
\normsub{g(Y)}{q} = \left(\expec{Y}{\abs{g(y)}^q}\right)^{1/q} = \left(\sum_{y\in Y}\mu(y)\abs{g(y)}^q\right)^{1/q}
\end{equation}
where $\mu(y)=\prod_{i=1}^n\mu(y_i)$ and each $\mu(y_i)$ takes the value according to Eq.~$\eqref{intro:bec:expec_distribution}$.
The matrix $\f{1}{2^n}\Piep$ in \cref{equation-qec-hc-introduction-d} acts as the role of $\mu$ in \cref{intro:bec:left_expanded}
as it reassigns the noise of the matrix when tracing out.
\end{remark}
\begin{remark}
For the classical BEC hypercontractivity bound, the parameter $c$ is always $c=1$ and this result is tight when $\ep\le1/2$.
In our work the tightness of the hypercontractivity results for the QEC remains unknown.
\end{remark}
\begin{remark}
The reader may wonder why we consider an inequality of the form in \cref{equation-qec-hc-introduction-d},
as opposed to the more natural hypercontractivity bound for the QEC expressed as
%\norm{2^{-n}\QECN(X)}_q\le\normthree{X}_p
\begin{equation}\label{equation-qec-hc-introduction-naive-d}
\left(2^{-n}\Tr\left[\left|\QECN(X)\right|^q\right]\right)^{1/q}\le \p{2^{-n}\Tr\lrb{\abs{X}^p}}^{1/p}.
\end{equation}
It's easy to check that
%$$\norm{2^{-n}\QECN(X)}_q=\left(2^{-n}\Tr\left[\left(\PiepN\QECN(X)\right)^q\right]\right)^{1/q}\le\left(2^{-n}\Tr\left[\PiepN\QECR^{\otimes n}(X)^q\right]\right)^{1/q}$$
$$\left(2^{-n}\Tr\left[\left|\QECN(X)\right|^q\right]\right)^{1/q}\le\left(2^{-n}\Tr\left[\PiepN\cdot\abs{\QECR^{\otimes n}(X)}^q\right]\right)^{1/q}.$$
Hence~\cref{equation-qec-hc-introduction-d} is stronger. Moreover,  \cref{equation-qec-hc-introduction-naive-d} fails to prove a tight bound on the communication complexity of common randomness generation considered in~\cref{section:CRG}.
We posit that the inequality in \cref{equation-qec-hc-introduction-d} will have broader applicability in a majority of scenarios,
as it is exactly the quantum generalization of its classical counterpart~\cite{7541363} as argued in Remark~\ref{rk:cq}. Furthermore, \cref{introduction-hc-informal} also implies an almost optimal hypercontractivity bound for the quantum depolarizing channel.
\end{remark}
%\begin{remark}
%Our result only applies to positive semidefinite matrices.
%For a general matrix $A$ that is not positive semidefinite,
%we can get an inequality by substituting $P$ with $\abs{A}$ in \cref{equation-qec-hc-introduction-d}.
%\end{remark}

\begin{remark}
    For the quantum erasure channel,
    our \cref{introduction-hc-informal} only proves hypercontractivity inequalities for $p\le 2$.
    The $p> 2$ case remains open.
    For quantum channels generated by primitive and reversible Lindblad generators,
    i.e., the depolarizing channel, a hypercontractivity inequality for $p$ and $q$ in $[1,2]$
    can be generalized to all parameter ranges.
    See Proposition 10 and Theorem 11 from \cite{beigi2020quantum}.
    However, for the quantum erasure channel considered in our work,
    \cref{introduction-hc-informal} does not seem to be able to imply hypercontractivity inequalities for the
    quantum erasure channel in the range $q > p > 2$.
    The barrier seems to be fundamental: For fixed $p$,
    we expect (and prove) that the left hand side of
    \cref{equation-qec-hc-introduction-d} is non-increasing over $\varepsilon$.
    However, for $q\ge 2$, numerical experiments show that this fails to be true.
    So our proof framework completely fails.
\end{remark}

\paragraph{Quantum Log-Sobolev Inequality for Bernoulli Entropy}
A key component in the proof of our hypercontractivity bound is a quantum log-Sobolev inequality for Bernoulli entropy. It is an extension of the standard quantum log-Sobolev inequality due to Kastoryano and Temme~\cite{kastoryano2013quantum}.
%, which is considered to be the single partite quantum log-Sobolev inequality.
%We believe this result is interesting in its own right and should have further applications.
%The full statement of our inequality is in \cref{hyper:lemma1}.
\begin{theorem}[informal version of \cref{hyper:lemma1}]
    Let $m,n$ be integers such that $m\le n$ and $X$ is a positive semi-definite operator.
    Let $\tr$ denote the normalized trace operator.
    For $q\in[1,2], \ep\in[0, 1]$, it holds that
    \begin{equation}\begin{split}
		\Ent_{\ep,[m],q}\lrb{X} \le
        2\wsum{m}{\sum_{k\in[m]\backslash S}}
        \p{\tr\lrb{\ps{X}{S}^{q}}-\tr\lrb{\ps{X}{\p{S\cup\set{k}}}^{q}}} \\
%\wsum{m}{\sum_{k\in[m]\backslash S}\tr {\left(\ps{X}{S}^{\f{q}{2}}-\ps{X}{\p{S\cup\set{k}}}^{\f{q}{2}}\otimes\id_{k}\right)^2}} \\
		+2\wsum{m}{\sum_{k\in[n]\backslash S}\p{\tr\lrb{\ps{X}{S}^q}-\tr\lrb{ {\left(\tr_{k}\lrb{\ps{X}{S}^{\f{q}{2}}}\right)^2}}}}.
    \end{split}\end{equation}
    where $X_T=\Tr_{T^c}X/2^{n-|T|}$ for $T\subseteq[n]$, and $\Ent_{\ep,[m],q}$ is the multipartite Bernoulli entropy defined as
    \begin{multline}
    \Ent_{\ep,[m],q}\lrb{X}=\sum_{S\subseteq[m]}(1-\epsilon)^{m-\abs{S}}\epsilon^{\abs{S}} \tr\lrb{\left(\tr_S X\right)^q\ln\left(\tr_S X\right)^q}\\
    -\left(\sum_{S\subseteq[m]}(1-\epsilon)^{m-\abs{S}}\epsilon^{\abs{S}}\tr\lrb{\left(\tr_S X\right)^q}\right)\ln\left(\sum_{S\subseteq[m]}(1-\epsilon)^{m-\abs{S}}\epsilon^{\abs{S}}\tr\lrb{\left(\tr_S X\right)^q}\right).
    \end{multline}
\end{theorem}

    % The definition of our multipartite entropy function $\Ent_{\ep,[m],q}\lrb{X}$ is \cref{pre:mpentdef}.
    We note that when $m=0$ and $q=2$, the first term on the right hand side disappears and thus the inequality boils down to the standard quantum log-Sobolev inequality \cref{hyper:theorem:log-sobolev}.

\paragraph{Common Randomness Generation}
We present an application of our hypercontractivity bound for the QEC in the study of common randomness generation, generalizing Guruswami and Radhakrishnan's result~\cite{guruswami_et_al:LIPIcs:2016:5845} to the quantum setting.
%Common randomness generation is a classic topic in information theory and distributed computing. It was first raised by Maurer~\cite{maurer1993secret} and Ahlswede and Csisz\'ar~\cite{ahlswede1998common}, which has lead numerous subsequent work. Readers may refer to the excellent survey~\cite{8863950} and the references therein. Recently,
%Guruswami and Radhakrishnan~\cite{guruswami_et_al:LIPIcs:2016:5845} have investigated the communication complexity common randomness generation when the parties share share unlimited but non-perfect correlation. They have considered both the binary symmetric channels and the binary erasure channels  and established tight tradeoffs among the communication complexity, min-entropy of the common randomness and errors for both cases. A key component of their proofs is hypercontrativity bounds for both channels.
%Devetak and Winter \cite{devetak2004distilling} introduced the problem entanglement-assisted common randomness generation,
%and studied the case of classical-quantum correlations. In~\cite{DY2023}, a subset of authors investigated the common randomness generation when the players share unlimited depolarized noisy EPRs states and also established tight tradeoffs among the classical communication complexity, min-entropy of the common randomness and errors. They also prove a tradeoff if quantum communication is allowed, which is not known to be tight. In this work, we consider the communication complexity of common randomness generation when the players share unlimited EPR states affected by quantum erasure channels.

\begin{theorem}[informal version of \cref{thm:CRG}]
  Let $\ep\in[0, 1]$ and $\ket{\Phi}=(\ket{00}+\ket{11})/\sqrt{2}$ be an EPR state.
  Suppose Alice and Bob share infinitely many copies of the state $(\id_2\otimes\QEC)(\Phi)$,
  and Alice sends classical messages to Bob.
  Then to produce a common random string of min-entropy at least $k$,
  Alice needs to send Bob at least
  \begin{equation}
    t\ge\frac{\ep k}{2}-o(k)
  \end{equation}
  bits of classical message.
\end{theorem}
This lower bound is tight up to a constant factor as it suffices for Alice to send $\ep k$ bits~\cite{guruswami_et_al:LIPIcs:2016:5845}.
%\subsection{Related work}
%\paragraph{Hypercontractivity}
%\cite{king2014hypercontractivity,10.1063/5.0056388} studied the hypercontractivity bounds for unital qubit channels.
%Their techniques are based on the log-Sobolev inequality,
%which in their paper is obtained by looking back at the hypercontractivity inequalities for the $1\le p\le 2\le q$ case.
%This case is simpler because there exists a simple inductive argument.
%Log-Sobolev inequalities are closely related to hypercontractivity and to some extend they can generally imply each other\cite{10.1215/S0012-7094-75-04237-4,Gross+2007+45+74}.
%
%Our work focuses on the quantum erasure channel.
%The hypercontractivity bounds for the classical erasure channel was studied in \cite{7541363}.
%In the one-qubit case, our result collapses to the classical case.
%Also, their result acts as a cornerstone of our result and we believe that
%the parameters for the quantum and classical hypercontractivity will be the same, i.e., $c=1$ in all cases.
%
%Quantum reverse hypercontractivity results have also been given \cite{10.1063/1.4933219, beigi2020quantum}.
%These bounds focus on the Schatten $p$-quasi norms when $p\in(0,1)$.
%The techniques involved are quite similar as they all concern log-Sobolev inequalities.

\subsection{Proof Overview}

The classical hypercontractiviy bound for BECs is proved in \cite{7541363} by an inductive argument on the length of the input $n$.
The base case is proved via an information-theoretic argument.
Then, for Boolean functions $g: \set{-1,1,*}^n\to\mathbb{R}$ with $n>1$ and 
the input $y\in\set{-1,1,*}^n$ of $g(y)$ is split into $y=y_1y_2$ where $y_1\in\set{-1,1,*}^{n_1}$ and $y_2\in\set{-1,1,*}^{n_2}$ and $n_1+n_2=n$.
Then by fixing $y_1$, we may apply the induction on the function $g(y_1, \cdot):\set{-1,1,*}^{n_2}\to\mathbb{R}$.
To generalize this inductive argument to quantum hypercontractivity, a natural approach is splitting the matrix into block matrices
and applying the induction hypothesis on each block. A key step is to compare
$$\norm{M}_p=\norm{\begin{pmatrix}X & Y\\Y^\dagger &Z \end{pmatrix}}_p \quad\text{ and }\quad m:=\norm{\begin{pmatrix}\normsub{X}{p}&\normsub{Y}{p}\\\normsub{Y}{p}&\normsub{Z}{p}\end{pmatrix}}_p$$
for positive semidefinite matrices. It is apparent that the inductive argument holds for the classical BEC hypercontractivity as $M$ is diagonal and thus $\norm{M}_p=m$ in the classical case. King~\cite{2003CMaPh.242..531K} proved a {\em norm-compression inequality} which states that $\norm{M}_p\geq m$ if $1\leq p\leq 2$ and $\norm{M}_p\leq m$ if $p\geq 2$.
Leveraging this fact, King~\cite{king2014hypercontractivity} generalized the inductive argument to the quantum setting and proved a hypercontractivity bound for the quantum depolarizing channel in the case that $1\le p\le 2\le q$.

The inductive argument also works for the hypercontractivity for the QEC, when $1\le p\le 2\le q$, which, therefore, induces a hypercontractivity bound for the QEC in this region\footnote{See \cref{appendix-proof-HC-1p2q}}. However, the inductive argument fails for the case $1\le p\le q \le 2$ as the direction of the norm compression inequality is opposite to what we expect.
To resolve this difficulty, we extend the framework in~\cite{OLKIEWICZ1999246,king2014hypercontractivity,beigi2020quantum,10.1063/5.0056388}.

We give a brief overview of this proof framework of quantum hypercontractivity. A key observation is the connection between quantum log-Sobolev inequality~\cite{kastoryano2013quantum} and quantum hypercontractivity. Such connections have also been explored in classical settings~\cite{POLYANSKIY2019108280,FOW2022general,kelman_et_al:LIPIcs.ITCS.2021.26}. A fundamental notion in the quantum log-Sobolev inequality is 
 the {\em entropy} function,
\begin{equation}\label{intro:eq:oldent}
	\Ent(X^2) = \tr \lrb{X^2\ln X^2} - \tr \lrb{X^2}\ln\tr \lrb{X^2},
\end{equation}
which captures the overall chaos of the system. It can be obtained by taking the derivative of the left-hand side of the hypercontractive inequality. The entropy is upper bounded via a quantum log-Sobolev inequality~\cite{kastoryano2013quantum} as follows:
\begin{equation}\label{intro:eq:oldlog}
    \Ent(X^2)=\tr \lrb{X^2\ln X^2} - \tr \lrb{X^2}\ln \tr \lrb{X^2} \le 2 \sum_{k=1}^n \p{\tr \lrb{X^2} - \tr \lrb{\p{\tr_k X}^2}}.
\end{equation}
Intuitively, the entropy is upper bounded by the correlation between each coordinate and the remaining system. The quantum hypercontractivity can be derived from the quantum log-Sobolev inequality via the quantum Gross' lemma~\cite{10.1215/S0012-7094-75-04237-4,king2014hypercontractivity}. Nonetheless, it only gives a quantum hypercontractive inequality for positive operators as the quantum log-Sobolev inequality only concerns about positive matrices. To extend it to general matrices, we can further employed Watrous' theorem, which asserts that the $\norm{\cdot}_{q\to p}$ norm of a quantum channel can be achieved by positive matrices~\cite{Wat05notes}. The overall proof is summarized in \cref{fig:overview}.

%\begin{figure}[hbt!]
%	\centering
%	\subfloat[King's work]{
%		\begin{tikzpicture}
%			\draw[draw=black] (-3, -3) rectangle ++(6, 8);
%			\node[draw, text width=150pt, align=center] (WT) at (0, 4) {Watrous' theorem~\cite{Wat05notes}};
%			\node[draw, text width=150pt, align=center] (LS) at (0, 2) {Log-Sobolev~\cite{kastoryano2013quantum}};
%			\node[draw, text width=150pt, align=center] (GL) at (0, 0) {Quantum Gross' lemma~\cite{10.1215/S0012-7094-75-04237-4}};
%			\node[draw, text width=150pt, align=center] (HC) at (0, -2) {HC for unital channels};
%			%\draw[double distance=2pt,->] (WT) -- node[above] {} (LS);
%			\draw[->] (WT) -- (LS);
%			\draw[->] (LS) -- (GL);
%			\draw[->] (GL) -- (HC);
%		\end{tikzpicture}
%	}\label{fig:overview:king}
%	\subfloat[\textbf{Our work}]{
%		\begin{tikzpicture}
%			\draw[draw=black] (-3, -3) rectangle ++(6, 8);
%			\node[draw, text width=150pt, align=center] (WT) at (0, 4) {Extended Watrous' theorem (\cref{hyper:lemma:psd})};
%			\node[draw, text width=150pt, align=center] (LS) at (0, 2) {Multipartite Log-Sobolev (\cref{hyper:lemma1})};
%			\node[draw, text width=150pt, align=center] (GL) at (0, 0) {Refined quantum Gross' lemma (\cref{hyper:lemma2})};
%			\node[draw, text width=150pt, align=center] (HC) at (0, -2) {HC for QEC (\cref{hyper:thm:hc})};
%			%\draw[double distance=2pt,->] (WT) -- node[above] {} (LS);
%			\draw[->] (WT) -- (LS);
%			\draw[->] (LS) -- (GL);
%			\draw[->] (GL) -- (HC);
%		\end{tikzpicture}
%	}\label{fig:overview:ours}
%	\caption{Proof Overview}\label{fig:overview}
%\end{figure}
%

\begin{figure}[hbt!]
	\centering
	\subfloat[Proof framework]{
		\begin{tikzpicture}
			\draw[draw=black] (-3, -3) rectangle ++(6, 8);
			\node[draw, text width=150pt, align=center] (WT) at (0, 4) {Log-Sobolev~\cite{kastoryano2013quantum}};
			\node[draw, text width=150pt, align=center] (LS) at (0, 2) {Quantum Gross' lemma~\cite{10.1215/S0012-7094-75-04237-4}};
			\node[draw, text width=150pt, align=center] (GL) at (0, 0) {Watrous' theorem~\cite{Wat05notes}};
			\node[draw, text width=150pt, align=center] (HC) at (0, -2) {HC for unital channels};
			%\draw[double distance=2pt,->] (WT) -- node[above] {} (LS);
			\draw[->] (WT) -- (LS);
			\draw[->] (LS) -- (GL);
			\draw[->] (GL) -- (HC);
		\end{tikzpicture}
	}\label{fig:overview:king}
	\subfloat[\textbf{Our work}]{
		\begin{tikzpicture}
			\draw[draw=black] (-3, -3) rectangle ++(6, 8);
			\node[draw, text width=150pt, align=center] (WT) at (0, 4) {Log-Sobolev for Bernoulli entropy (\cref{hyper:lemma1})};
			\node[draw, text width=150pt, align=center] (LS) at (0, 2) {Refined quantum Gross' lemma (\cref{hyper:lemma2})};
			\node[draw, text width=150pt, align=center] (GL) at (0, 0) {Watrous' theorem~\cite{Wat05notes}};
			\node[draw, text width=150pt, align=center] (HC) at (0, -2) {HC for QEC (\cref{hyper:thm:hc})};
			%\draw[double distance=2pt,->] (WT) -- node[above] {} (LS);
			\draw[->] (WT) -- (LS);
			\draw[->] (LS) -- (GL);
			\draw[->] (GL) -- (HC);
		\end{tikzpicture}
	}\label{fig:overview:ours}
	\caption{Proof Overview}\label{fig:overview}
\end{figure}
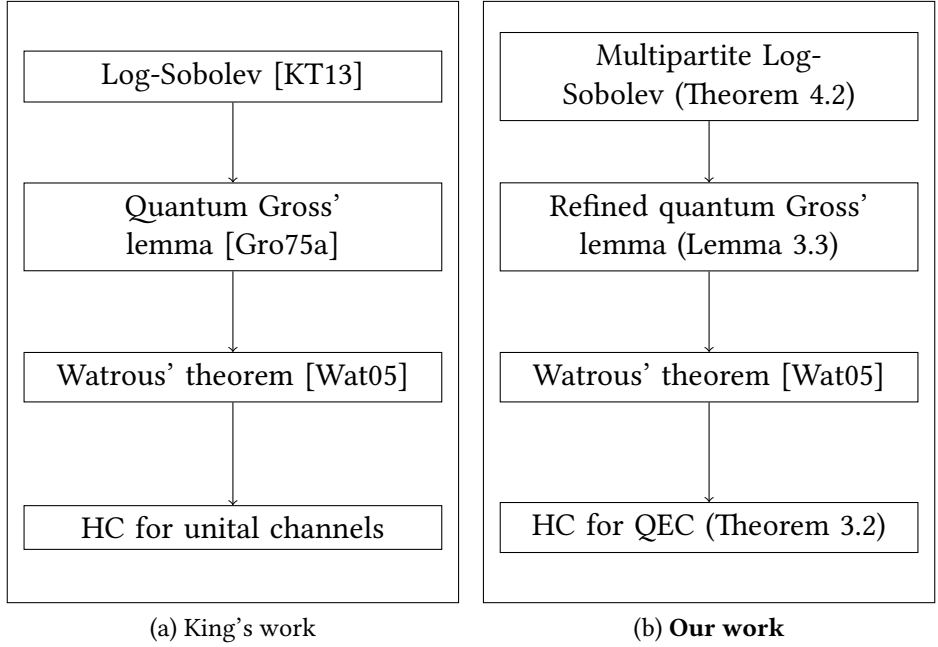

To prove the hypercontractive inequality for the QEC, we adopt the same framework. However, we need to introduce several new ingredient, which is also depicted in \cref{fig:overview}.

The major difficulty is that the quantum log-Sobolev inequality is insufficient for the hypercontractivity of QEC, because the size of the system considered in Eq.~\eqref{intro:eq:oldlog} is unchanged, while the qubits may be erased and the size of the system decreases when considering the QEC. To capture the ``erasing'' feature of the QEC, we introduce the notion of {\em multipartite Bernoulli entropy}. Let $X$ be a $2^n\times 2^n$ positive definite matrix viewed as an operator acting on $n$ qubits, and $m\leq n$.
Recall the {\em multipartite Bernoulli entropy} defined as
\begin{samepage}
\begin{multline}
\Ent_{\ep,[m],q}\lrb{X}=\sum_{S\subseteq[m]}(1-\epsilon)^{m-\abs{S}}\epsilon^{\abs{S}} \tr\lrb{\left(\tr_S X\right)^q\ln\left(\tr_S X\right)^q}\\
-\left(\sum_{S\subseteq[m]}(1-\epsilon)^{m-\abs{S}}\epsilon^{\abs{S}}\tr\lrb{\left(\tr_S X\right)^q}\right)\ln\left(\sum_{S\subseteq[m]}(1-\epsilon)^{m-\abs{S}}\epsilon^{\abs{S}}\tr\lrb{\left(\tr_S X\right)^q}\right).
\end{multline}
\end{samepage}
The multipartite Bernoulli entropy of $X$ captures the {\em expected} entropy of a random system obtained by removing each qubit system in the first $m$ qubits systems with probability $\ep$, independently.  
It is a natural extension of the entropy in \cref{intro:eq:oldent}, which is the case that $m=0$ and $q=2$. We establish a  quantum log-Sobolev inequality for Bernoulli entropy, which enables us to upper bound the multipartite Bernoulli entropy in terms of the the expected correlation between each coordinate and the remaining system, where the expectation is over the random erasure of the system.
\begin{equation*}
\begin{split}
	\Ent_{\ep,[m],q}\lrb{X} \le
    2\wsum{m}{\sum_{k\in[m]\backslash S}}
        \p{\tr\lrb{\ps{X}{S}^{q}}-\tr\lrb{\ps{X}{\p{S\cup\set{k}}}^{q}}} \\
	+2\wsum{m}{\sum_{k\in[n]\backslash S}\p{\tr\lrb{\ps{X}{S}^q}-\tr\lrb{ {\left(\tr_{k}\lrb{\ps{X}{S}^{\f{q}{2}}}\right)^2}}}}.
\end{split}
\end{equation*}
We also note that when $m=0$ and $q=2$ our quantum log-Sobolev inequality degenerates to the quantum log-Sobolev inequality above as \cref{intro:eq:oldlog}.  It it is worth to note that for $q\neq 2$, the inequality resembles, to some extent, the nonlinear log-Sobolev inequality considered in \cite{POLYANSKIY2019108280,9996873}.

The existing quantum Gross' lemma is also insufficient to derive a hypercontractivity for the QEC from the quantum log-Sobolev inequality. To address it, we prove a refined quantum Gross' lemma connecting the Schatten $p$ norm of an operator and the Schatten $p$
norm of its reduced operator via exploiting the majorization argument between a matrix's eigenvalues and its diagonal entries.

%The final issue about the above argument is that we only consider positive semidefinite matrices. Watrous' theorem fails when proving a hypercontractivity bound for QECs, because the left hand side of the hypercontractive inequality for QECs is not a norm of a quantum channel. Hence, we extend Watrous' argument~\cite{Wat05notes} to QECs again using majorization argument between eigenvalues and diagonal entries and prove that the it suffices to consider positive matrices as well.

Finally, we conclude our result by employing Watrous' theorem~\cite{Wat05notes},
which implies that the above argument for positive semi-definite matrices is sufficient
to prove a QEC hypercontractivity for general matrices.

\subsection{Discussion and open problems}

We prove a hypercontractive inequality for the product of quantum erasure channels.
Our hypercontractivity is a natural extension
of the hypercontractivity bound for the classical binary erasure channel. To prove the hypercontractivity, we establish a quantum log-Sobolev inequality for Bernoulli entropy, which is interesting in its own right. As an application, we prove an almost tight lower bound on the classical communication complexity of common randomness generation assisted by erased-noisy EPR states. This work raises several open problems for future work.
\begin{enumerate}
  \item For parameters $2\le p\le q < \infty$, the hypercontractivity of the quantum erasure channel remains unknown.
        Our technique fails because the derivative $g^\prime(t)$ in the proof of \cref{hyper:thm:hc} fails to be non-increasing for $q\ge 2$.
        More technically, \cref{lem:doublystho} holds only for $q\ge 2$.
        Thus the hypercontractivity for $2\le p\le q < \infty$ remains open.

	\item The tightness of hypercontractivity bound for the QEC is not clear. The hypercontractivity bound for classical BECs~\cite{7541363} indicates that the parameter $c$ in our result could be pushed to $1$. We leave this to future work.
	
	\item
  To our knowledge, hypercontractive inequalities for tensored channels, i.e., the exact tensorization property,
  has been proved only for qubit channels~\cite{montanaro2010quantum,kastoryano2013quantum,king2014hypercontractivity}.
  For qudit channels, so far there are only lower bounds from bounds on the log-Sobolev constant~\cite{Temme_2014,MHSFW16nonunital}
  (See the discussion in \cite{beigi2020quantum}).
  Is it possible to establish the tensorization property for qudit channels?
  The current approach fails even for the easy case $1\leq p\leq 2\leq q$,
  because the norm-compression inequality does not hold for $t\times t$ block matrices whenever $t\geq 3$ \cite{AUDENAERT2006155}.
  Hence, we need new techniques to establish quantum hypercontractive inequalities for the product of qudit channels.
	
	\item There have been several works investigating the hypercontractivity bound for quantum channels with respect to the fixed states of the channels over this space~\cite{Bardet2022-bg,beigi2020quantum}.
    %However, they only imply the known hypercontractivity for the quantum depolarizing channels.
    For hypercontractivity bounds with respect to a non-fixed state,
    Bardet and Rouz\'e~\cite{Bardet2022-bg} have shown some no-go results: a strong decoherence-free-log-Sobolev inequality does not hold for a quantum Markov semigroups that is neither primitive nor unital.
    Can we say more about a general quantum hypercontractive inequality in weighted non-commutative $L_p$ spaces with respect to a non-fixed state?

	\item Log-Sobolev inequalities are powerful tools to study the mixing time of a random walk on a fixed graph. Recently, there have been several works investigating random walks on a random graph \cite{10.1214/07-AOP358,10.1214/EJP.v17-1705,rsa.20539}. The entropy in the classical log-Sobolev inequality for Bernoulli entropy captures the average entropy of the random subgraphs of a given graph, where each vertex is deleted with probability $\ep$. Thus, it is interesting to see whether our classical  log-Sobolev inequality for Bernoulli entropy can provide better analysis on the mixing time of a random walk on a random graph?
\end{enumerate}

\subsection{Organization}
%%%%%%%%%%%%%% need to be refined (TODO)
In \cref{section:Preliminary}, we introduce some preliminary concepts, basic mathematical foundations, notations, and discuss the erasure channel.
\cref{section:hc} presents the proof of our hypercontractivity inequality. \Cref{sec:ls} presents the proof of quantum log-Sobolev inequality for Bernoulli entropy. Subsequently, in \cref{section:CRG} we demonstrate an application of our results to derive a lower bound for common randomness generation.
%Finally, we conclude our paper and highlight some open problems in \cref{section:Conclusion}.

 % @dyj

\begin{Anonymous}
\subsection*{Acknowledgment}
We thank the anonymous reviewers for their careful reading of our manuscript and their many insightful comments and suggestions.
We thank Mark M. Wilde for pointing out their work~\cite{10005080,10161613,nuradha2023fidelitybased}.
We thank Pei Wu for pointing out their work~\cite{10.1145/3564246.3585205}.
We are grateful to Xinyuan Zhang and Xiaoyu Chen for their discussion about log-Sobolev inequalities and entropy functions.
We thank Jingcheng Liu and Yitong Yin for the helpful discussion.
This research was supported by National Natural Science Foundation of China (Grant No. 62332009, 12347104),  Innovation Program for Quantum Science and Technology (Grant No. 2021ZD0302901), NSFC/RGC Joint Research Scheme (Grant no. 12461160276) and Natural Science Foundation of Jiangsu Province (No. BK20243060) and the New Cornerstone Science Foundation.

\end{Anonymous}

\section{Preliminary}\label{section:Preliminary}
In this paper, we will use  $\ln x$ to express the natural logarithm of
$x$, which is based on the mathematical constant  $e \approx 2.7182818$.
We denote $[n] = \{1,2,\cdots,n\}$ and for the special case $[0]=\emptyset$.
Besides, unless explicitly stated, we assume $\ep \in [0,1]$. For complex number $z$, we use $\overline{z}$ to denote its complex conjugate.
%The symbol $\wsum{m}{f(S), \ep}$ is an abbreviation of $\sum_{S \subseteq [m]} (1-\ep)^{m-|S|} \ep^{|S|} f(S)$. In cases where the context is clear, we omit the explicit notation $\ep$ and simply denote it as $\wsum{m}{f(S)}$.

\paragraph{Linear algebra}

%When the size of the identity matrix can be inferred from the context, we omit the subscript $n$.
%Even though it may lead to some symbol abuse, $\id_k$ sometimes used to refer to the identity matrix corresponding to the $k$-th qubit.
Let $\id_n$ be the $n \times n$ size identity matrix and we use $\id$ to represent $\id_2$.
We let $\Matrix_{n\times m}$ be the set of all $n \times m$ size complex-valued matrices.
We use $\Matrix_n$ to denote all $n\times n$ square matrices and $\Herm_{n}$ to denote all $n \times n$ Hermitian matrices.
For $X \in \Matrix_{n}$, we define its trace as $\Tr [X] \defeq \sum_{i = 1}^n X_{ii}$. Additionally, its {\em normalized} trace is defined as $\tr [X] \defeq \frac{1}{n} \Tr [X]$.

$X \in \Matrix_{n}$ is a positive semi-definite (PSD) matrix if $X$ is Hermitian and $\mathbf{x}^{\dagger} X \mathbf{x} \geq 0$ holds for all vectors $\mathbf{x} \in \mathbb{C}^n$, where $\mathbf{x}^{\dagger}$ is the complex conjugate transpose of $\mathbf{x}$. We may write $X \geq 0$ for a PSD matrix $X$. If $X \geq 0$ satisfies $\Tr [X] = 1$, we say $X$ is a density operator. The set of all PSD matrices and density operators is denoted as $\PSD_n$ and $\Density_n$.

For a vector $\mathbf{x} = (x_1,\cdots,x_n)^T \in \mathbb{C}^n$ and $1 \leq p < \infty$, we use $\norm{\mathbf{x}}_p = (\sum_{i=1}^n |x_i|^p)^{1/p}$ to denote its $p$-norm.
For $X \in \Matrix_{n \times m}$, the Schatten $p$-norm is defined to be $\norm{X}_p = (\Tr \lrb{\abs{X}^p})^{1/p}$. We define the {\em normalized} Schatten $p$-norm as $\normthree{X}_p=\br{\tr\Br{\abs{X}^p}}^{1/p}$, where $\abs{X} \defeq \sqrt{X^\dagger X}$ and $X^\dagger$ is the complex conjugate transpose of $X$.
%Here $p$ and $q$ are called to be \Holder\ conjugate.

Let $N = 2^n$ for positive integer $n$ and $X \in \Matrix_{N}$. All the rows and columns are indexed by the subsets of $[n]$. The partial trace of $X$ with respect to $S \subseteq [n]$ is defined as a square matrix $\Tr_S [X]$ with indices ranging over the subsets of $S^c$, typically, 
\begin{equation}
  \Tr_S [X] = \sum_{s \in \set{0,1}^S} \p{I_{S^c} \otimes \bra{s}} X \p{I_{S^c} \otimes \ket{s}}.
\end{equation}
%Intuitively, the partial trace is the matrix obtained by observing the matrix $X$ through a certain subset.
Besides, the {\em normalized} partial trace is defined as $\tr_S [X]\defeq 2^{-\abs{S}} \cdot \Tr_S \lrb{X}$. For simplicity, we denote $\tr_S [X]$ by $\ps{X}{S}$ and $\tr_k [X]$ by $\ps{X}{\set{k}}$.

Schatten $p$-norms admit \Holder's inequality (\cite[Eq. 1.174]{watrous2018theory} ):
Let $A,B\in\Matrix_{n \times m}$ and $p, p^*$ be positive real numbers satisfying $1/p+1/p^*=1$, we have
\begin{equation}\label{preliminary-Holder-2}
\normsub{A}{p} = \max\set{\abs{\Tr [C^\dagger A]}: C\in\Matrix_{n \times m}, \normsub{C}{p^*}\le 1}
\end{equation}
which implies
\begin{equation}\label{preliminary-Holder-1}
\qquad \abs{\Tr [B^\dagger A]} \leq \normsub{A}{p} \cdot \normsub{B}{p^*}.
\end{equation}
$p^*$ above is called the \emph{\Holder\ conjugate} of $p$.
\Holder's inequality can also be generalized to two sequences of matrices.
\begin{lemma}[Matrix H\"{o}lder's inequality]\cite[ Theorem 2.6]{shebrawi_albadawi_2013}\label{preliminary-Holder-3}
  Let $A_i, B_i(i=1, 2, \dots, m)$ be two sequences of matrices
and $p, q$ be positive real numbers satisfying $1/p+1/q=1$. We have
$$\abs{\Tr\left[\sum_{i=1}^mA_i^\dagger B_i\right]}\le\left(\Tr\left[\sum_{i=1}^m\abs{A_i}^p\right]\right)^{1/p}\left(\Tr\left[\sum_{i=1}^m\abs{B_i}^q\right]\right)^{1/q}.$$

Specifically, if $A_i$ and $B_i$ are positive semi-definite for all $i$, then
$$\Tr\left[\sum_{i=1}^mA_iB_i\right]\le\left(\Tr\left[\sum_{i=1}^mA_i^p\right]\right)^{1/p}\left(\Tr\left[\sum_{i=1}^mB_i^q\right]\right)^{1/q}.$$
\end{lemma}
\begin{proof}
  Let $A=\bigoplus_{i=1}^mA_i$ and $B=\bigoplus_{i=1}^mB_i$.
  %Said equivalently,
  %$$A = \begin{pmatrix}A_1 & & \\&\ddots&\\& &A_m \end{pmatrix}, B = \begin{pmatrix}B_1 & & \\&\ddots&\\& &B_m \end{pmatrix}$$
  By the H\"{o}lder's inequality for Schatten $p$-norm
$$\abs{\Tr\left[\sum_{i=1}^mA_i^\dagger B_i\right]}=\abs{\Tr\left[A^\dagger B\right]}\le\normsub{A}{p}\normsub{B}{q}=\left(\Tr\left[\sum_{i=1}^m\abs{A_i}^p\right]\right)^{1/p}\left(\Tr\left[\sum_{i=1}^m\abs{B_i}^q\right]\right)^{1/q}.$$
\end{proof}

\paragraph{Majorization}

A doubly stochastic matrix $X \in \mathbb{R}^{n \times n}$ is a non-negative matrix with all row sums and column sums equal to $1$. For $\mathbf{x} \in \mathbb{R}^n$, we let $\mathbf{x}^{\downarrow}$ be the vector by rearranging the values of $\mathbf{x}$ in decreasing order.
\begin{definition}[Majorization]
  We say $\mathbf{x}$ majorizes $\mathbf{y}$ (denoted as $\mathbf{x} \succeq \mathbf{y}$) if
  $$\sum_{i=1}^n \mathbf{x}_i = \sum_{i=1}^n  \mathbf{y}_i\text{ and }\sum_{i=1}^k \mathbf{x}^{\downarrow}_{i} \geq \sum_{i=1}^k \mathbf{y}^{\downarrow}_{i}$$
  for all $k \in [n]$. Two key properties of majorization are listed below:
\end{definition}

\begin{theorem}[Hardy-Littlewood-P\'{o}lya Theorem]\cite[page 33]{marshall1979inequalities}\label{pre:HLPthm} % Theorem B.2.
Let $\mathbf{x}, \mathbf{y} \in \mathbb{R}^n$, then $\mathbf{x} \preceq \mathbf{y}$ if and only if there exists a doubly stochastic matrix $D$ such that $\mathbf{x} = D\mathbf{y}$.
\end{theorem}

\begin{theorem}[Schur-Horn Theorem]  \cite[page 300]{marshall1979inequalities}\label{pre:shthm} % Theorem B.1.
For $X \in \Herm_{n}$, let $\mathbf{diag} = (X_{11}, \cdots, X_{nn})^T$ be the diagonal entries , $\bda = (\lambda_1(X), \cdots, \lambda_n(X))^T$ where $\lda_i(X)$ denotes the $i$-th biggest eigenvalue of $X$. Then $\mathbf{diag} \preceq \bda$.
\end{theorem}

According to \cref{pre:HLPthm}, the majorization relation in \cref{pre:shthm} relates to a doubly stochastic matrix $D$. The following lemma implicates that $D$ has an intriguing property.

\begin{lemma}\label{pre:shthm-lemma}
Under the setting of \cref{pre:shthm}, let $X=U\Sigma U^\dagger$  be a spectral decomposition of $X$ with diagonal matrix $\Sigma$ satisfying $\Sigma_{i,i}=\lambda_i(X)$. Then it holds that $\mathbf{diag} = D \bda$, where $D_{ij} = |U_{ij}|^2$ is a doubly stochastic matrix.
\end{lemma}

\begin{proof}
    For $\forall i\in [n]$, $X_{i,i}=\sum_{j=1}^n U_{ij}\Sigma_{jj}\overline{U_{ij}}=\sum_{j=1}^n \abs{U_{ij}}^2\Sigma_{jj}=\sum_{j=1}^n D_{ij}\bda_{j}$.	
\end{proof}

\paragraph{Entropy}

In order to define our entropy notations used in our proof of hypercontractivity inequality, we start with a weighted relative entropy, which unifies the three definitions of relative entropies below.
Note that in this work the entropy is usually defined on quantum operators,
as a quantum analog of the entropy for functions in the context of log-Sobolev inequalities.
This is in contrast to the standard quantum information theory definition of entropy primary used for quantum states.
\begin{definition}[Weighted Relative Entropy]
Given an index set $I$, with a sequence of non-negative numbers $(w_i)_{i \in I}$ satisfying $\sum_{i \in I} w_i = 1$ and PSD matrices $(X_i)_{i \in I}$, we define the \emph{weighted relative entropy} as
$$\text{Entropy}\lrb{(w_i,X_i)_{i\in I}} \defeq \sum_{i \in I} w_i \cdot \tr \lrb{X_i \ln X_i} - \p{\tr \lrb{\sum_{i \in I} w_i \cdot X_i }} \ln \p{\tr\lrb{ \sum_{i \in I} w_i \cdot X_i}}.$$
\end{definition}

Generally, one can easily check $\text{Entropy}\lrb{(w_i,X_i)_{i \in I}} \geq 0$ with Jensen's inequality. In this paper, we consider three special forms of the weighted relative entropy.

%We will omit the description regarding $I$ and  directly list $w_i$ and $X_i$ in the order of their corresponding relationship.

\begin{definition}[Two-point Relative Entropy]\label{def:twopointrel}
\noindent For any $a,b>0,\ep\in[0,1]$,
\begin{align*}
\Ent_{\ep}\lrb{a, b} &\defeq \text{Entropy}\left[(1-\ep,a),(\ep,b)\right]\\
&=(1-\ep)a\ln a+\ep b\ln b-\br{\br{1-\ep}a+\ep b}\ln \br{\br{1-\ep}a+\ep b}. \nonumber
\end{align*}
\end{definition}

This definition is used fruitfully when exploring isoperimetry properties over general probability distributions. See \cite{BARTHE2000419, bob98entropy} for more details.
\begin{definition}[One-Partite Relative Entropy]
\noindent For $X \geq 0$ and $q \geq 1$, we define the entropy as:
\begin{align*}\Ent_{q}\lrb{X} &\defeq \text{Entropy}\left[\p{1,X^q}\right]\\
&=\tr \lrb{X^q\ln X^q}-\tr \lrb{X^q}\ln\tr \lrb{X^q}.\nonumber
\end{align*}
\end{definition}

This definition is related to the quantum relative entropy. Assume $X=2^n \phi$ where $\phi$ is a density operator. Then $\Ent_1[X]= D\p{\phi \parallel \f{\id_{2^n}}{2^n}}$, where $D\p{\phi \parallel \f{\id_{2^n}}{2^n}}$ is the quantum relative entropy between $\phi$ and the maximally mixed state.

This particular quantity can be traced back to the work presented in \cite{3f89ed44-5628-3112-93f5-de6a2b2ab43a}. Besides, this term of relative entropy will inevitably arise when we employ the similar derivative techniques used in \cite{10.1063/5.0056388,beigi2020quantum} to prove the quantum hypercontractive inequality.

Here we denote $\Ent_1[X^q]$ by $\Ent_q[X]$ to keep it consistent with the definition below.

\begin{definition}[Bernoulli Multi-Partite Relative Entropy]\label{pre:mpentdef}
\noindent For $X \in \PSD_{2^n}$, $S \subseteq [n]$ and $q \geq 1$, we define the entropy as:
\begin{eqnarray}
&&\Ent_{\ep, S, q}\lrb{X} \defeq \text{Entropy}
\left[
\p{(1-\ep)^{|S|-|T|}\ep^{|T|},\ps{X}{T}^q}_{T \subseteq S}
%	    ((1-\ep)^{|S|-|T|}\ep^{|T|})_{T \subseteq S}, (\ps{X}{T}^q)_{T \subseteq S}
\right]\label{ent}\\
&=&\sum_{T\subseteq S}(1-\ep)^{|S|-|T|}\ep^{|T|}\tau\Br{\ps{X}{T}^q\ln\br{\ps{X}{T}^q}}-\nonumber\\
&&\br{\tau\Br{\sum_{T\subseteq S}(1-\ep)^{|S|-|T|}\ep^{|T|}\ps{X}{T}^q}}\ln\br{\tau\Br{\sum_{T\subseteq S}(1-\ep)^{|S|-|T|}\ep^{|T|}\ps{X}{T}^q}}.\nonumber
\end{eqnarray}
\end{definition}

This definition is equivalent to \cref{intro:entdef}. The term \emph{Bernoulli} is derived from the internal utilization of the Bernoulli distribution. We identify it as \emph{multi-partite} due to the inclusion of all partial traces of $X$ within the expression. This quantity reflects, to some degree, the relative entropy of each constituent party of $X$. It's important to note that this denotes a natural extension of the one-partite relative entropy.

\paragraph{Quantum Erasure Channel}

A quantum channel $\Phi$ is a linear map in $L(\Matrix_{n},\Matrix_{n'})$ with the following properties: (i) completely positive : $\forall \text{ PSD matrix } X \in \Matrix_{n n'}$, we have $\p{\Phi_n \otimes \id_{n'}}(X) \geq 0$; (ii) trace preserving : $\forall X \in \Matrix_{n}$, $\Tr\lrb{\Phi(X)} = \Tr \lrb{X}$. A useful equation is \cite[Eq. 2.70]{watrous2018theory}: for  $\Phi: \Matrix_n\to\Matrix_m$, $X\in\Matrix_n$ and $Y\in\Matrix_m$, $\Tr(\Phi(X)\cdot Y) = \Tr(X \cdot \Phi^\dagger(Y))$.

To introduce quantum erasure channels, we first define an expanding operator $\QECR$ in $L(\Matrix_{2}, \Matrix_{3})$:
$$\QECR(X) \defeq \begin{pmatrix} X &  \\ & \tr(X) \end{pmatrix}.$$
The single-qubit quantum erasure channel (QEC) $\QECR_{\ep}$ is defined as
$$\QECR_{\ep}(X) \defeq \Pi_\ep  \cdot \QECR(X)\text{, where } \Pi_\ep \defeq \begin{pmatrix} 1-\ep & & \\ & 1-\ep& \\ &&2\ep \end{pmatrix}.$$
Here $\Pi_{\ep}$ can be viewed as an erasure noise matrix. The definition above can be extended to the $n$-qubit case as follows.
Recall that Pauli matrices
\[\set{\sigma_0=\begin{pmatrix}1 & 0\\0 & 1\end{pmatrix}, \sigma_1=\begin{pmatrix}0 & 1\\1 & 0\end{pmatrix}, \sigma_2=\begin{pmatrix}0 & -i\\i & 0\end{pmatrix}, \sigma_3=\begin{pmatrix}1 & 0\\0 & -1\end{pmatrix}}\]
form a base of $\Matrix_{2}$ and thus $\set{\sigma_0, \sigma_1, \sigma_2, \sigma_3}^{\otimes n}$ form a base of $\Matrix_{2^n}$.
For any $X = \sum_{x \in \mathbb{Z}_4^n} \left(\widehat{X}_x \bigotimes_{i=1}^n \sigma_{x_i}\right)$,
we define $\QECR^{\otimes n}$ in $L(\Matrix_{2^n}, \Matrix_{3^n})$ as
$$\QECR^{\otimes n}(X) \defeq \sum_{x \in \mathbb{Z}_4^n} \widehat{X}_x \left(\bigotimes_{i=1}^n \QECR\left(\sigma_{x_i}\right)\right)$$
and the $n$-qubit quantum erasure channel is defined as $\QECR_\ep^{\otimes n}(X)=\Pi_\ep^{\otimes n} \cdot \QECR^{\otimes n}(X)$. When the context is clear, we omit the symbols $\otimes n$ and use $\QECR$, $\Pi_\ep$ and $\QECR_\ep$ to denote the $n$-qubit case. An $n$-qubit QEC induces a norm, which we refer to as the $(\ep, q)$-norm in $\Matrix_{2^n}$.
\begin{definition}\label{def:epqnorm}
	For any integer $n>0,\ep\in[0,1]$ and $q\geq 1$, the $(\ep,q)$-norm of $X\in\Matrix_{2^n}$ is
	\[\normthree{X}_{\ep, q} \defeq \p{2^{-n} \cdot \Tr \lrb{ \Pi_\ep \cdot |\QECR(X)|^q}}^{\f{1}{q}}.\]
\end{definition}
To see that it is a norm, it suffices to verify that  $\normthree{X}_{\ep,q}=\normthree{ (\Pi_\ep^{1/2q})^{\otimes n}|\QECR^{\otimes n}(X)|(\Pi_\ep^{1/2q})^{\otimes n}}_q$.

For a quantum channel with Kraus representation $\Phi(\rho)=\sum_i X_i\rho X_i^\dagger$,
we define the conjugate map
$$\Phi^\dagger(\rho) = \sum_i X_i^\dagger\rho X_i.$$
In particular, the conjugate map of the quantum erasure channel is $\QEC^\dagger: \Matrix_3\to\Matrix_2$:
\begin{equation}
  \QEC^\dagger(\rho) = (1-\ep)\cdot\begin{pmatrix}1 & 0 & 0 \\0 & 1 & 0 \end{pmatrix}  \rho
   \begin{pmatrix} 1 & 0 \\ 0 & 1 \\ 0 & 0 \end{pmatrix}+ \ep\bra{\Errsym}\rho\ket{\Errsym}\cdot\id_2.
\end{equation}

Note that the set of all possible outputs of a quantum erasure channel forms a subspace of $\Matrix_{3^n}$, where all elements in the subspaces commute with $\PiepN$.
Throughout this paper, we will only encounter matrices in this subspace of $\Matrix_{3^n}$.
Such matrices have the form of
\begin{equation}\label{app:def:goodform}
  Y = \sum_{S\subseteq[n]}\ps{Y}{S}\otimes\Err^{\otimes S}
\end{equation}
where each $\ps{Y}{S}$ lies in the subspace $\Matrix_{2^{n-\abs{S}}}$. The following two facts can be easily verified by the definition.
\begin{fact}
  For any operator $X\in\Matrix_{2^n}$, $\QECRN(X)$ is of \goodform.
\end{fact}
\begin{fact}
Suppose $Y\in\Matrix_{3^n}$ is of \goodform.
Then $Y$ commutes with $\PiepN$ for all $0\le\ep\le 1$.
\end{fact}

 % @ofn

\section{Hypercontractivity}\label{section:hc}
In this section we prove the hypercontractivity for the product of quantum erasure channels. Recall that for $X\in \Matrix_{2^n}$,
$$
    \normthree{X}_{\ep, q} \defeq \p{2^{-n} \cdot \Tr \lrb{ \PiepN \cdot |\QECRN(X)|^q}}^{\f{1}{q}} = \p{\wsum{n}{\normthree{\tr_S X}_q^q}}^{\f{1}{q}}.
$$

%We first employ the argument similar to \cite{Wat05notes} along with a majorization reasoning to argue that it suffices to consider positive semidefinite matrices.

The following lemma states that it suffices to restrict our argument to positive semi-definite matrices,
as the supremum of $\normthree{X}_{\ep,q}/\normthree{X}_p$ is achieved by a positive semi-definite matrix.
This is implicitly implied by Watrous' paper~\cite{Wat05notes} by noting that $\normthree{X}_{\ep,q} = \normthree{\p{\PiepN}^{1/q}\cdot\QECRN(X)}_q$
and that the map $X\mapsto\p{\PiepN}^{1/q}\cdot\QECRN(X)$ is completely positive. For the sake of completeness, the proof of the following lemma is provided in \cref{app:psd}.

\begin{restatable}[Watrous' theorem~\cite{Wat05notes}]{lemma}{hyperlemmapsd}\label{hyper:lemma:psd}
  For $1\le p\le q,\ep\in [0,1]$, we have
  $$
%    	    \sup_{X\in \M_{2^n\times 2^n},X\neq 0} \f{\p{2^{-n}\Tr \lrb{ \Pi_\ep \cdot \abs{D(X)}^q}}^{\f{1}{q}}}{\p{\tr \lrb{\abs{X}^p}}^{\f{1}{p}}} \le \sup_{X\in \PSD_{2^n},X\neq 0} \f{\p{2^{-n}\Tr \lrb{ \Pi_\ep \cdot \abs{D(X)}^q}}^{\f{1}{q}}}{\p{\tr \lrb{\abs{X}^p}}^{\f{1}{p}}}
   \sup_{X\in \Matrix_{2^n},X\neq 0} \f{\normthree{X}_{\ep,q}}{\normthree{X}_p} = \sup_{X\in \PSD_{2^n},X\neq 0} \f{\normthree{X}_{\ep,q}}{\normthree{X}_{p}}.
  $$
\end{restatable}

The main theorem of this section is the following.
\begin{theorem}[Hypercontractivity for the product of QECs]\label{hyper:thm:hc}
	For any $X\in \Matrix_{2^n}$, $ q\ge p\ge 1, \ep\in[0, 1],c\ge 1$ satisfying $1-\ep \leq \p{\f{p-1}{q-1}}^c$, it holds that
	%    \begin{equation}\label{hyper:eq:thm}
		%    	\p{2^{-n}\Tr \lrb{ \Pi_\ep \cdot \abs{D(X)}^q}}^{\f{1}{q}} \le \p{\tr \lrb{\abs{X}^p}}^{\f{1}{p}}
		%    \end{equation}
	\begin{equation}\label{hyper:eq:thm:main}
		\normthree{X}_{\ep,q} \le \normthree{X}_p,
	\end{equation}
	whenever either of the following two cases holds.
	
	\begin{itemize}
		\item \textbf{Case 1}: $c=1$ and $1\le p\le 2\le q$.
		\item \textbf{Case 2}: $c=2$ and $1\le p\le q\le 2$.
	\end{itemize}	
\end{theorem}

\begin{proof}
	\textbf{Case 1} has a simple inductive proof by using the norm-compression inequality~\cite{2003CMaPh.242..531K}.
	For the sake of completeness, we provide a proof in \cref{appendix-proof-HC-1p2q}.
\textbf{Case} 2 is much more challenging. The rest of this section is devoted to this case.
	
    By \cref{hyper:lemma:psd}, we can assume $X$ to be positive semi-definite. We note that if $q=1$, our conclusion holds trivially. Assume $q>1$ for the following arguments. Fix $X$ and $p\in[1,2]$. Recall that
    \[\normthree{X}_{\ep,q}=\p{2^{-n} \cdot \Tr \lrb{ \PiepN \cdot |\QECRN(X)|^q}}^{\f{1}{q}}~\mbox{and}~\normthree{X}_p=\p{2^{-n}\cdot \Tr\lrb{\abs{X}^p}}^{\f{1}{p}}.\]
    Introduce a parameter $t$ and let $\f{p-1}{q-1}=e^{-t/c}$ and $\ep=1-e^{-t}$.
    It is easy to see that \cref{hyper:eq:thm:main} holds with equality at $t=0$, since in this case $q=p$ and $\ep=0$.
    %To prove our theorem,
    %we are going to show for any $q_0\in [1,2]$ and constant $c\ge 8/q_0^2$,
    %the left-hand side of \cref{hyper:eq:thm:main} is non-increasing over $t$,
    %as long as $q\in [q_0,2]$.
	
	Let $g(t)=\ln\,\normthree{X}_{\ep,q}$, calculating the derivative of $g(t)$ over $t$ yields
	\begin{align*}
%		g'(t) &= \f{q-1}{c\cdot q^2}\cdot \f{1}{2^{-n}\Tr \lrb{ \Pi_\ep \cdot D(X)^q}}\cdot \\
		g'(t) &= \f{q-1}{c\cdot q^2}\cdot \f{1}{\normthree{X}_{\ep,q}^q}\cdot \\
		&\qquad\qquad\lrb{\Ent_{\ep,[n],q}\lrb{X} - \f{c\cdot q}{q-1}\wsum{n}{\sum_{k\in [n]\backslash S}\p{\tr \lrb{ \ps{X}{S}^q} - \tr\lrb{ \ps{X}{\p{S\cup \{k\}}}^q }}}}
	\end{align*}
	and recall that $\Ent_{\ep,[n],q}\lrb{X}$ is defined as
	$$
	    \wsum{n}{\tr\lrb{\ps{X}{S}^q\ln\ps{X}{S}^q}} - \p{\wsum{n}{\tr\lrb{\ps{X}{S}^q}}}\ln\p{\wsum{n}{\tr\lrb{\ps{X}{S}^q}}}.
	$$
	To show $g'(t)\le 0$, we need to prove
	\begin{equation}\label{hyper:eq:after_derivative}
	    \Ent_{\ep,[n],q}\lrb{X} \le  \f{c\cdot q}{q-1}\wsum{n}{\sum_{k\in [n]\backslash S}\p{\tr \lrb{ \ps{X}{S}^q} - \tr\lrb{ \ps{X}{\p{S\cup \{k\}}}^q }}}.
	\end{equation}
    Combining \cref{hyper:lemma1} with $m=n$ and \cref{hyper:lemma2} for $c^\prime=4/q^2$, we have
    \begin{align*}
      \Ent_{\ep,[n],q}\lrb{X} &\le \p{2+\f{2}{q-1}} \wsum{n}{\sum_{k\in [n]\backslash S}\p{\tr \lrb{\ps{X}{S}^q} - \tr \lrb{\ps{X}{\p{S\cup \{k\}}}^q}}}.
    \end{align*}
    Note that
    $$2+\f{2}{q-1} = \f{2\cdot q}{q-1} = \f{c\cdot q}{q-1}.$$
    This completes the proof.

\end{proof}

%\textbf{The following proof SHOULD be rewritten}
The following lemma, to some extent, resembles quantum Gross' lemma~\cite{10.1215/S0012-7094-75-04237-4,king2014hypercontractivity},
and the quantum Stroock-Varopoulos inequality~\cite{10.1063/1.4933219,beigi2020quantum,bardet2017estimating}.
Notice that, these inequalities are essential in the proof of the hypercontractivity of unital qubit channels~\cite{10.1215/S0012-7094-75-04237-4,king2014hypercontractivity,10.1063/5.0056388}.

\begin{restatable}[Refined Gross' lemma]{lemma}{hyperlemmaii}\label{hyper:lemma2}
	For $X\in\PSD_{2^n}$ and  $q\in (1,2]$,
	\begin{align}
%		\tr \p{X^{\f{q}{2}}-\p{\tr_n X}^{\f{q}{2}}\otimes \id}^2 &\le c'\cdot \f{q}{q-1}\p{\tr \lrb{X^q} - \tr \lrb{\p{\tr_n X}^q}}\label{eqn1}\\
		\tr \lrb{X^q} - \tr\lrb{ \p{\tr_n \lrb{X^{\f{q}{2}}}}^2} &\le \f{1}{q-1}\p{\tr \lrb{X^q} - \tr \lrb{\p{\tr_n X}^q}}\label{eqn2}.
	\end{align}
\end{restatable}

\begin{remark}
This lemma differs to the quantum Stroock-Varopoulos inequality from \cite{beigi2020quantum} 
 in that the right hand side is not a Dirichlet form, thus is new as far as we can tell.
\end{remark}

The proof of \cref{hyper:lemma2} is built on the following Lemma.

\begin{lemma}\label{lem:doublystho}
	Let $q\in (1,2], \bda\in\reals_{\geq 0}^{2^n}$ and $D$ be a $2^n\times 2^n$ doubly stochastic matrix. It holds that
	\begin{equation}\label{ineq:new_goal}
		\norm{\bda}_q^q - \norm{D \bda^{\f{q}{2}}}_2^2 \le \f{1}{q-1}\p{\norm{\bda}_q^q -\norm{D \bda}_q^q}.
	\end{equation}
\end{lemma}

\begin{proof}[Proof of \cref{hyper:lemma2}]
   We will reduce \eqref{eqn2} to \eqref{ineq:new_goal}.
%    To see \eqref{eqn1}, expanding the left-hand side, we have:
%    		    \tr \p{X^{\f{q}{2}}-\p{\tr_n X}^{\f{q}{2}}\otimes \id}^2 &\le \f{c'\cdot q}{q-1}\p{\tr \lrb{X^q} - \tr \lrb{\p{\tr_n X}^q}}\\
%		    \tr \lrb{X^q} - \tr \p{\tr_n \lrb{X^{\f{q}{2}}}}^2 &\le \f{c'\cdot q}{2(q-1)}\p{\tr \lrb{X^q} - \tr \lrb{\p{\tr_n X}^q}}
%    \begin{align*}
%    	\tr \p{X^{\f{q}{2}}-\p{\tr_n X}^{\f{q}{2}}\otimes \id}^2  &= \tr \lrb{X^q} + \tr \lrb{\p{\tr_n X}^q} - 2 \tr \lrb{X^{\f{q}{2}} \cdot \p{\tr_n X}^{\f{q}{2}}\otimes \id} \\
%    	&\le 2\tr \lrb{X^q} - 2 \tr \lrb{X^{\f{q}{2}} \cdot \p{\tr_n X}^{\f{q}{2}}\otimes \id}
%    \end{align*}
%
%    It is sufficient to show
%    \begin{equation}\label{ineq:final2}
%    	\tr \lrb{X^q} - \tr \lrb{X^{\f{q}{2}} \cdot \p{\tr_n X}^{\f{q}{2}}\otimes \id} \le \f{c'\cdot q}{2(q-1)}\p{\tr \lrb{X^q} - \tr \lrb{\p{\tr_n X}^q}}
%    \end{equation}

    Without loss of generality, we assume $\tr_n X$ is diagonal from now on. Otherwise we could replace $X$ by $\p{V\otimes \id}^\dagger  X\p{V\otimes \id}$ where $V$ is a unitary to diagonalize $\tr_n X$ without changing any term in the equation above.

    Let $\diag(X)$ be the column vector of $X$'s diagonal entries and $\bda(X)$ be the column vector of $X$'s eigenvalues. Both vectors are sorted in non-increasing order. Then there exists a unitary $U$ s.t. $X=U\Sigma U^\dagger$, where $\diag(\Sigma)=\bda(X)$. According to \cref{pre:shthm} and \cref{pre:shthm-lemma}, we have $\bda(X) \succeq\diag(X)$ and $\diag(X)=D\bda(X)$ for some doubly stochastic matrix $D$, where $D$ depends only on $U$ and is independent of $\Sigma$. This property is crucial since we can deduce $\diag\p{X^{\f{q}{2}}}=D \bda\p{X^{\f{q}{2}}}=D \bda(X)^{\f{q}{2}}$ with the same $D$, where we use $v^q$ to represent the vector $(v_1^q,\ldots,v_n^q)^T$ for column vector $v=(v_1,\ldots,v_n)^T$.

    Besides, since $\tr_n X$ is diagonal, we know $\bda(\tr_n X \otimes \id)=\diag(\tr_n X \otimes \id)=M \diag(X) = MD \bda(X)$ and $\diag\p{\tr_n \lrb{X^{\f{q}{2}}} \otimes \id}=M \diag\p{X^{\f{q}{2}}} = MD \bda(X)^{\f{q}{2}}$, where $M=\f{1}{2} \cdot \id_{2^{n-1}}\otimes\lrb{\begin{matrix}1\ 1\\1\ 1\end{matrix}}$. As an intuition, $M\diag(X)$ calculates average values of each two corresponding entries of $\diag(X)$. We use $D'$ to denote $M\cdot D$. Note that $D'$ is also doubly stochastic.

    With all the definitions above, we see that
    \begin{align}
    	2^n\tr \lrb{X^q} &= \norm{\bda(X)}_q^q \\
    	2^n\tr \lrb{  \p{\tr_nX}^q} = 2^n\tr \lrb{\p{\tr_n X}^q\otimes \id_2} &= \norm{D' \bda(X)}_q^q \\
        2^n \tr \lrb{\p{\tr_n \lrb{X^{\f{q}{2}}}}^2} = \norm{\tr_n \lrb{X^{\f{q}{2}}} \otimes \id}_2^2 &\ge \norm{\diag\p{\tr_n \lrb{X^{\f{q}{2}}} \otimes \id}}_2^2 = \norm{D' \bda(X)^{\f{q}{2}}}_2^2
%    	\tr \lrb{X^{\f{q}{2}} \cdot \p{\tr_n \lrb{X}}^{\f{q}{2}}\otimes \id} &= \f{1}{2^n} \p{D\bda(X)^{\f{q}{2}}}^T \cdot \p{D' \bda(X)}^{\f{q}{2}}
    \end{align}
    
    Therefore,
    \begin{align*}
    	2^n\p{\tr \lrb{X^q} - \tr\lrb{ \p{\tr_n \lrb{X^{\f{q}{2}}}}^2}} &\le \norm{\bda(X)}_q^q - \norm{D' \bda(X)^{\f{q}{2}}}_2^2 \\
    	2^n\p{\tr \lrb{X^q} - \tr \lrb{\p{\tr_n X}^q}} &= \norm{\bda(X)}_q^q -\norm{D' \bda(X)}_q^q.
    \end{align*}
    According to \cref{lem:doublystho}, we have
    $$
		\norm{\bda(X)}_q^q - \norm{D' \bda(X)^{\f{q}{2}}}_2^2 \le \f{1}{q-1}\p{\norm{\bda(X)}_q^q -\norm{D' \bda(X)}_q^q}
    $$
    since $D'$ is a doubly stochastic matrix, thus we conclude that
    $$
        \tr \lrb{X^q} - \tr\lrb{ \p{\tr_n \lrb{X^{\f{q}{2}}}}^2} \le \f{1}{q-1}\p{\tr \lrb{X^q} - \tr \lrb{\p{\tr_n X}^q}}.
    $$

    \end{proof}

    \begin{proof}[Proof of \cref{lem:doublystho}]
%    Here we reach the end of the first part. In the second part of the proof, we will show
%
%    holds for arbitrary double stochastic matrix $D\in \mathcal{M}_{n\times n}$, arbitrary non-negative vector $\bda\in \mathbb{R}^n$, $\forall q\in (1,2]$ and $c'=4/q^2$. It is easy to see our goals follows from \cref{ineq:new_goal}.
%
%    \paragraph{Outline of the second part of proof} This part of proof is dedicated to prove Eq~\ref{ineq:new_goal}. The doubly stochastic matrix here allows us to split the original complicated problem into smaller and simpler ones, and for each small problem, the property of related function is very good for us to derive the desired result.
%
%    \paragraph{Second part of proof}
    First notice that
    \begin{equation*}
    \sum_{ij}D_{ij}\lda_j^q = \sum_j\p{\sum_i D_{ij}}\lda_j^q = \sum_j\lda_j^q = \normsub{\bda}{q}^q.
    \end{equation*}
    Thus \cref{ineq:new_goal} is equivalent to
    \begin{equation*}
    \sum_i\sum_jD_{ij}\lda_j^q + \sum_i\p{\sum_jD_{ij}\lda_j^\frac{q}{2}}^2\le\f{1}{q-1}\p{\sum_i\sum_jD_{ij}\lda_j^q-\sum_i\p{\sum_jD_{ij}\lda_j}^q}.
    \end{equation*}
    It suffices to prove for every $i$,
    \begin{equation*}
    \sum_jD_{ij}\lda_j^q + \p{\sum_jD_{ij}\lda_j^\frac{q}{2}}^2\le\f{1}{q-1}\p{\sum_jD_{ij}\lda_j^q-\p{\sum_jD_{ij}\lda_j}^q}
    \end{equation*}
    holds, which again suffices to prove
    \begin{equation*}%\label{ineq:new_goal_small}
        \sum_{i=1}^n\alpha_i\lda_i^q - \p{\sum_{i=1}^n \alpha_i\lda_i^{\f{q}{2}}}^2 \le \f{1}{q-1}\p{\sum_{i=1}^n \alpha_i\lda_i^q - \p{\sum_{i=1}^n\alpha_i\lda_i}^q}
    \end{equation*}
    for an arbitrary non-negative vector $\{\alpha_i\}_{i=1}^n$ satisfying $\sum_{i=1}^n\alpha_i=1$ and an arbitrary non-negative vector $\{\lda_i\}_{i=1}^n$ and $q\in (1,2]$.

    %%%%% new proof %%%%%%%%%%
    Treating $\set{\alpha_i}_{i=1}^n$ as a distribution over $[n]$ and considering the normalized $p$-norms
    \begin{equation}
      \normthree{\lda}_p = \p{\sum_{i}\alpha_i\abs{\lda_i}^p}^{1/p} = \p{\expec{i}{\abs{\lda_i}^p}}^{1/p}.
    \end{equation}
    Our goal can be tidied into
    \begin{equation}
      \normthree{\lda}_{1}^q \le (2-q)\normthree{\lda}_{q}^q + (q-1)\normthree{\lda}_{q/2}^q.
    \end{equation}
    We first use the \Holder's inequality for probability spaces.
    Let $p=q/(2-q)\ge 1$ and $p^*=q/(2(q-1))$ be $p$'s \Holder\ conjugate, we obtain
    \begin{multline}
      \normthree{\lda}_1 = \expec{}{\lda_i} = \expec{i}{\lda_i^{2-q}\cdot \lda_i^{q-1}} \\
      \le\normthree{\lda^{2-q}}_p\cdot\normthree{\lda^{q-1}}_{p^*} =\p{\expec{i}{\lda_i^q}}^{(2-q)/q}\cdot\p{\expec{i}{\lda_i^{q/2}}}^{2(q-1)/q}
      = \normthree{\lda}_q^{2-q}\cdot\normthree{\lda}_{q/2}^{q-1}.
    \end{multline}
    We now use the following inequality:
    Let $0\le s \le 1$ and $a, b\ge 0$, then
    \begin{equation}\label{equation:mult-to-sum}
      a^sb^{1-s} \le sa + (1-s)b.
    \end{equation}
    This follows simply by the concavity of the $\log$ function:
    $$\log a^sb^{1-s} = s \log a + (1-s)\log b \le \log\p{sa + (1-s)b}.$$
    Then we apply \cref{equation:mult-to-sum} and get
    \begin{equation}
      \normthree{\lda}_q^{2-q}\cdot\normthree{\lda}_{q/2}^{q-1} \le (2-q)\normthree{\lda}_q + (q-1)\normthree{\lda}_{q/2}.
    \end{equation}
    Finally, by convexity of the function $x\mapsto x^q$,
    \begin{equation}
      \normthree{\lda}_1^q \le \p{(2-q)\normthree{\lda}_q + (q-1)\normthree{\lda}_{q/2}}^q\le(2-q)\normthree{\lda}_q^q + (q-1)\normthree{\lda}_{q/2}^q.
    \end{equation}

\end{proof}

\section{Variable Multipartite Log-Sobolev Inequality}\label{sec:ls}
In this section, we will prove a quantum log-Sobolev inequality for multipartite Bernoulli entropy, which plays a crucial role in the proof of the hypercontractivity for quantum erasure channels. The standard quantum log-Sobolev inequality is due to Kastoryano and Temme~\cite{kastoryano2013quantum}.

\begin{theorem}[Quantum log-Sobolev inequality]\cite{kastoryano2013quantum}\label{hyper:theorem:log-sobolev}
	For $A\in \PSD_{2^n}$, it holds that
	$$
	\Ent_2 \lrb{A} = \tr \lrb{A^2\ln A^2} - \tr \lrb{A^2}\ln \tr \lrb{A^2} \le 2 \sum_{k=1}^n \p{\tr \lrb{A^2} - \tr \lrb{\p{\tr_k \lrb{ A}}^2}}.
	$$
\end{theorem}

In \Cref{hyper:theorem:log-sobolev}, the dimension of the system is fixed. We now state our main result of the variable multipartite quantum log-Sobolev inequality, where the dimension of the system varies subject to a random partial trace.
Recall the definition
\begin{multline*}
\Ent_{\ep,[m],q}\lrb{X}=\sum_{S\subseteq[m]}(1-\epsilon)^{m-\abs{S}}\epsilon^{\abs{S}} \tr\lrb{\left(\tr_S X\right)^q\ln\left(\tr_S X\right)^q}\\
-\left(\sum_{S\subseteq[m]}(1-\epsilon)^{m-\abs{S}}\epsilon^{\abs{S}}\tr\lrb{\left(\tr_S X\right)^q}\right)\ln\left(\sum_{S\subseteq[m]}(1-\epsilon)^{m-\abs{S}}\epsilon^{\abs{S}}\tr\lrb{\left(\tr_S X\right)^q}\right).
\end{multline*}

\begin{theorem}[Variable Quantum log-Sobolev inequality]\label{hyper:lemma1}
		Let $m,n$ be integers such that $m\le n$ and $X\in\PSD_{2^n}$.
		For $q\in[1,2], \ep\in[0, 1]$, it holds that
		\begin{equation}\label{hyper:eq:induction}\begin{split}
		\Ent_{\ep,[m],q}\lrb{X} \le
        2\wsum{m}{\sum_{k\in[m]\backslash S}}
        \p{\tr\lrb{\ps{X}{S}^{q}}-\tr\lrb{\ps{X}{\p{S\cup\set{k}}}^{q}}} \\
%\wsum{m}{\sum_{k\in[m]\backslash S}\tr {\left(\ps{X}{S}^{\f{q}{2}}-\ps{X}{\p{S\cup\set{k}}}^{\f{q}{2}}\otimes\id_{k}\right)^2}} \\
		+2\wsum{m}{\sum_{k\in[n]\backslash S}\p{\tr\lrb{\ps{X}{S}^q}-\tr\lrb{ {\left(\tr_{k}\lrb{\ps{X}{S}^{\f{q}{2}}}\right)^2}}}}.
		\end{split}\end{equation}
    where $X_T=\Tr_{T^c}X/2^{n-|T|}$ for $T\subseteq[n]$, and $\Ent_{\ep,[m],q}$ is the multipartite Bernoulli entropy defined as
    \begin{multline}\label{intro:entdef}
    \Ent_{\ep,[m],q}\lrb{X}=\sum_{S\subseteq[m]}(1-\epsilon)^{m-\abs{S}}\epsilon^{\abs{S}} \tr\lrb{\left(\tr_S X\right)^q\ln\left(\tr_S X\right)^q}\\
    -\left(\sum_{S\subseteq[m]}(1-\epsilon)^{m-\abs{S}}\epsilon^{\abs{S}}\tr\lrb{\left(\tr_S X\right)^q}\right)\ln\left(\sum_{S\subseteq[m]}(1-\epsilon)^{m-\abs{S}}\epsilon^{\abs{S}}\tr\lrb{\left(\tr_S X\right)^q}\right).
    \end{multline}
\end{theorem}
%	\begin{restatable}[Multipartite quantum log-Sobolev inequality]{theorem}{hyperlemmai}\label{hyper:lemma1}
%		Let $m,n$ be integers s.t. $m\le n$ and $X\in\PSD_{2^n}$, $\Tr X >0$.
%		For $q\in(1,2], \ep\in[0,1/2]$,
%		\begin{equation}\label{hyper:eq:induction}\begin{split}
%		\Ent_{\ep,[m],q}\lrb{X} \le\wsum{m}{\sum_{k\in[m]\backslash S}\tr {\left(\ps{X}{S}^{\f{q}{2}}-\ps{X}{\p{S\cup\set{k}}}^{\f{q}{2}}\otimes\id_{k}\right)^2}} \\
%		+2\wsum{m}{\sum_{k\in[n]\backslash S}\p{\tr\lrb{\ps{X}{S}^q}-\tr {\left(\tr_{k}\lrb{\ps{X}{S}^{\f{q}{2}}}\right)^2}}}
%		\end{split}\end{equation}
%    \end{restatable}
%    When $m=0$, it boils down to the standard log-Sobolev inequality \cref{hyper:theorem:log-sobolev}.

To the best of our knowledge, even the classical counterpart of~\cref{hyper:lemma1} is also unknown. We write it down for readers who are interested in classical log-Sobolev inequalities and their applications.

\begin{cor}[Variable Classical log-Sobolev inequality]
	For $f:\set{0,1}^n\rightarrow\reals_{\geq 0}$, integers $m\leq n$, $T\subseteq[n]$,we define $f_T:\set{0,1}^{|T|}\rightarrow\reals_{\geq 0}$ to be $f_T(x)=\expec{y\sim\set{0,1}^{|T^c|}}{f(x,y)}$, where $x$ and $y$ are placed on $T$ and $T^c$, respectively.
For $\ep\in[0, 1], q\in[1, 2]$, define
	\begin{eqnarray*}
		\Ent_{\ep,[m],q}\lrb{f}=&&\sum_{S\subseteq[m]}\br{1-\ep}^{m-|S|}\ep^{|S|}\expec{}{f_{S^c}^q\ln\br{f_{S^c}^q}}\\
		&-&\br{\sum_{S\subseteq[m]}\br{1-\ep}^{m-|S|}\ep^{|S|}\expec{}{f_{S^c}^q}\ln\br{\sum_{S\subseteq[m]}\br{1-\ep}^{m-|S|}\ep^{|S|}\expec{}{f_{S^c}^q}}}.
	\end{eqnarray*}
	It holds that 
	\begin{eqnarray*}
		\Ent_{\ep,[m],q}\lrb{f}&\leq&2\sum_{S\subseteq[m]}\br{1-\ep}^{m-|S|}\ep^{|S|}\sum_{k\in[m]\setminus S}\br{\expec{}{f^q_{S^c}}-\expec{}{f^q_{\br{S\cup\set{k}}^c}}}\\
		&+&2\sum_{S\subseteq[m]}\br{1-\ep}^{m-|S|}\ep^{|S|}\sum_{k\in[n]\setminus S}\br{\expec{}{f^q_{S^c}}-\expec{}{\br{f_{S^c}^{q/2}}_{\br{S\cup\set{k}}^c}^2}}.
	\end{eqnarray*}
	
\end{cor}

We are now ready to prove \Cref{hyper:lemma1}.

\begin{proof}[Proof of~\cref{hyper:lemma1}]
    Define
    \begin{align*}
        \co\p{X,m}&\defeq \wsum{m}{\sum_{k\in[m]\backslash S}}
        \p{\tr\lrb{\ps{X}{S}^{q}}-\tr\lrb{\ps{X}{\p{S\cup\set{k}}}^{q}}},
        %{\left(\ps{X}{S}^{\f{q}{2}}-\ps{X}{\p{S\cup\set{k}}}^{\f{q}{2}}\otimes\id_{k}\right)^2}}
        \\
		\ct\p{X,m}&\defeq\wsum{m}{\sum_{k\in[n]\backslash S}\p{\tr\lrb{\ps{X}{S}^q}-\tr \lrb{{\left(\tr_{k}\lrb{\ps{X}{S}^{\f{q}{2}}}\right)^2}}}}.
	\end{align*}
	It suffices to prove
  \begin{equation}\label{hyper:eq:eplrb}
	    \Ent_{\ep,[m],q}\lrb{X}\le 2\cdot\co\p{X,m}+2\cdot \ct\p{X,m}.
  \end{equation}
  We will prove by induction on $m$.
  Recall $[0]=\emptyset$ which implies $\co\p{X, 0}=0$.
  Thus our inequality in the base case of  $m=0$ is equivalent to the standard quantum log-Sobolev inequality \cref{hyper:theorem:log-sobolev} where $A=X^{\f{q}{2}}$.

	%\paragraph{Induction step}
  Now suppose $m\ge1$ and \cref{hyper:eq:eplrb} holds
  for all integers in $[m-1]$.
	Recall the definition of two-point relative entropy $\Ent_\ep\Br{\cdot,\cdot}$ in~\cref{def:twopointrel}. We expand the left-hand side
%	\begin{align*}
%	    \wsum{m}{\tr\lrb{\ps{X}{S}^q\ln\ps{X}{S}^q}} &= (1-\ep)\cdot\wsum{m-1}{\tr\lrb{\ps{X}{S}^q\ln\ps{X}{S}^q}} \\
%	    &\qquad\qquad + \ep\cdot\wsum{m-1}{\tr\lrb{\ps{X}{S\cup\set{m}}^q\ln\ps{X}{\p{S\cup\set{m}}}^q}} \\
%	    \wsum{m}{\tr\lrb{\ps{X}{S}^q}} &= (1-\ep)\cdot\wsum{m-1}{\tr\lrb{\ps{X}{S}^q}} \\
%	    &\qquad\qquad+ \ep\cdot\wsum{m-1}{\tr\lrb{\ps{X}{\p{S\cup\set{m}}}^q}}
%	\end{align*}
%
%    then
	\begin{equation*}\label{hyper:eq:induction2}\begin{split}
	  	&\Ent_{\ep,[m],q}\lrb{X}
%	        =\,&\wsum{m}{\tr\lrb{\ps{X}{S}^q\ln\ps{X}{S}^q}}- \wsum{m}{\tr\lrb{\ps{X}{S}^q}}\ln\wsum{m}{\tr\lrb{\ps{X}{S}^q}} \\
	        =(1-\ep)\cdot \Ent_{\ep,[m-1],q}\lrb{X} + \ep\cdot \Ent_{\ep,[m-1],q}\lrb{\ps{X}{\set{m}}} \\
	        &\qquad + \Ent_\ep\lrb{\wsum{m-1}{\tr\lrb{\ps{X}{S}^q}}, \wsum{m-1}{\tr\lrb{\ps{X}{\p{S\cup\set{m}}}^q}}}.
	\end{split}\end{equation*}
	By the induction hypothesis, we have
	\begin{align*}
	    &(1-\ep)\cdot\Ent_{\ep,[m-1],q}\lrb{X}+\ep\cdot\Ent_{\ep,[m-1],q}\lrb{\ps{X}{\set{m}}} \\
     =&2(1-\ep)\cdot\co\p{X,m-1}+2(1-\ep)\cdot \ct\p{X,m-1}
		  + 2\ep\cdot\co\p{\ps{X}{\set{m}},m-1}+2\ep\cdot \ct\p{\ps{X}{\set{m}},m-1} \\
     =& 2\co\p{X,m} + 2\ct\p{X,m}- 2\wsumr{m-1}{m}{\p{\tr\lrb{\ps{X}{S}^q} - \tr\lrb{\ps{X}{\p{S\cup\set{m}}}^q}}}
	\end{align*}
  since
  $$(1-\ep)\cdot\ct\p{X,m-1} + \ep\cdot\ct\p{\ps{X}{\set{m}},m-1} = \ct\p{X,m},$$
  $$(1-\ep)\cdot\co\p{X,m-1} + \ep\cdot\co\p{\ps{X}{\set{m}},m-1} = \co\p{X,m} - \wsumr{m-1}{m}{\p{\tr\lrb{\ps{X}{S}^q} - \tr\lrb{\ps{X}{\p{S\cup\set{m}}}^q}}}.$$
%	\begin{align}
%	    &\quad\begin{aligned}\label{hyper:eq:ind1}
%	    \Ent_{\ep,[m-1],q}\lrb{X} &\le\wsum{m-1}{\sum_{k\in[m-1]\backslash S}\tr\left(\ps{X}{S}^{\f{q}{2}}-\ps{X}{\p{S\cup\set{k}}}^{\f{q}{2}}\otimes\id_{k}\right)^2} \\
%		&\qquad+\PA\cdot\wsum{m-1}{\sum_{k\in[n]\backslash S}\tr\lrb{\ps{X}{S}^q}-\tr\left(\tr_{k}\lrb{\ps{X}{S}^{\f{q}{2}}}\right)^2}
%		\end{aligned} \\
%        &\begin{aligned}\label{hyper:eq:ind2}
%		&\Ent_{\ep,[m-1],q}\lrb{\ps{X}{\set{m}}}\le\wsum{m-1}{\sum_{k\in[m-1]\backslash S}\tr\left(\ps{X}{\p{S\cup\set{m}}}^{\f{q}{2}}-\ps{X}{\p{S\cup\set{k,m}}}^{\f{q}{2}}\otimes\id_{k}\right)^2} \\
%		&\qquad\qquad\qquad+\PA\cdot\wsum{m-1}{\sum_{k\in[n]\backslash S}\tr\lrb{\ps{X}{\p{S\cup\set{m}}}^q}-\tr\left(\tr_{k}\lrb{\ps{X}{\p{S\cup\set{m}}}^{\f{q}{2}}}\right)^2}
%		\end{aligned}
%	\end{align}
Finally with \cref{claim:eq}, we obtain
	\begin{align*}
		\Ent_{\ep,[m],q}\lrb{X} \le 2\cdot \co\p{X,m} + 2\cdot \ct\p{X,m},
	\end{align*}
	which completes the proof of the induction step.

	%To summarize, the proof of Lemma~\ref{hyper:lemma1} is finished by induction.
\end{proof}
\begin{claim}\label{claim:eq}
		\begin{equation}\label{hyper:eq:ind3}\begin{split}
			&\Ent_\ep\lrb{\wsum{m-1}{\tr\lrb{\ps{X}{S}^q}, \wsum{m-1}{\tr\lrb{\ps{X}{\p{S\cup\set{m}}}^q}}}} \\
      &\le2\wsumr{m-1}{m}{\p{\tr\lrb{\ps{X}{S}^q} - \tr\lrb{\ps{X}{\p{S\cup\set{m}}}^q}}}.
	\end{split}\end{equation}
\end{claim}
%The proof of \cref{claim:eq}  is postponed to \cref{app:induction}.
\begin{proof}%[Proof of \cref{claim:eq}]
  By \cref{app:lemma:partialtracenormd} and the fact that $1\le q\le 2$, we have for each $S$,
  $$\tr\lrb{\ps{X}{(S\cup\set{m})}^q} \le \tr\lrb{\ps{X}{S}^q} \le 16\tr\lrb{\ps{X}{(S\cup\set{m})}^q}.$$
  Thus we are able to apply \cref{app:lemma:techlemma}, to get
    \begin{align*}
        &\Ent_\ep\lrb{\wsum{m-1}{\tr\lrb{\ps{X}{S}^q}}, \wsum{m-1}{\tr\lrb{\ps{X}{(S\cup\set{m})}^q}}} \\
        &\le 2(1-\ep) \wsum{m-1}{\p{\tr\lrb{\ps{X}{S}^q}- \tr\lrb{\ps{X}{(S\cup\set{m})}^q}}} \\
        &= 2 \wsumr{m-1}{m}{\p{\tr\lrb{\ps{X}{S}^q}- \tr\lrb{\ps{X}{(S\cup\set{m})}^q}}}. \\
    \end{align*}

%    We finish our proof with the following observation:
%    \begin{align*}
%      &(1-\ep) \wsum{m-1}{\p{\tr\lrb{\ps{X}{S}^q}- \tr\lrb{\ps{X}{(S\cup\set{m})}^q}}} \\
%      &= \co(X,m) - (1-\ep) \cdot \co(X,m-1) - \ep \cdot \co(\ps{X}{\set{m}}, m-1).
%    \end{align*}
\end{proof}

\section{Common Randomness Generation}\label{section:CRG}
In this section we give a lower bound on the communication cost of \textit{common randomness} (CR) generation,
in the case where Alice and Bob share maximally entangled states affected by the quantum erasure channel as free resource.

CR generation refers to the task of producing a common random string $X$ accessible for both Alice and Bob.
In our work we restrict the actions of Alice and Bob to local quantum operations and one-way classical communication from Alice to Bob.
We measure the amount of randomness of a random variable $X$ by its min-entropy.
\begin{definition}[Min-Entropy]
Let $R$ be a random variable with distribution $\mu$. We define the min-entropy of $R$ to be
\begin{equation}
\Hmin{R}\defeq\min_{r\in R}\log\left(1/\mu(r)\right).
\end{equation}
We assume $\log(1/0)=+\infty$ in this definition.
\end{definition}
Let $\ket{\Phi}=\p{\ket{00}+\ket{11}}/\sqrt{2}$ be an EPR state and write $\Phi=\ketbra{\Phi}$.
If Alice and Bob share maximally entangled states,
they can apply the same measurements $\set{M_0=\ketbra{0}, M_1=\ketbra{1}}$ on the shared entanglement and get exactly identical results.
Thus they can generate CR without any classical communication.
However, EPRs may suffer various noise in reality. Thus, it is natural to investigate CR generation with non-perfect quantum correlation.

Recall that the quantum erasure channel $\QEC: \Matrix_2\to\Matrix_3$ for $0\le\ep\le1$ is
$$\QEC(X) = \Piep\QECR(X) = (1-\ep)\cdot X+\ep\Tr(X)\cdot\Err.$$
%From now on we will write $$\Err=\id_2 \otimes\ketbra{1}$$
%for the error part.
%When it is clear from the context, in the above definition, we prefer to write
%$$\QEC(\rho) \defeq (1-\ep)\cdot \rho+\ep\cdot\tr\rho\cdot \Err$$
We study the communication cost of CR generation,
when Alice and Bob share EPR states while Bob's quantum states are affected by a quantum erasure channel.
Note that the quantum erasure channel is not symmetric regarding to maximally entangled states:
\begin{equation}
(\QEC\otimes\id)(\Phi)\neq(\id\otimes\QEC)(\Phi),
\end{equation}
so we make explicit here that the noise is applied to Bob's part,
i.e., they share the state $(\id\otimes\QEC)(\Phi)$
where Alice has access to the left side and Bob has access to the right side.
This is more interesting because it is Alice who sends one-way information to Bob.
Otherwise if the noise is applied to Alice's part,
Alice can simply tell Bob which qubits are faulty.

Suppose Alice and Bob share $n$ copies of the state $(\id\otimes\QEC)(\Phi)$,
and their goal is to produce a CR of min-entropy at least $k$,
with one-way classical communication from Alice to Bob.
Alice's strategy can be represented by a POVM $\set{X_{a, \pi}}_{a, \pi}$,
where $X_{a, \pi}\in\PSD_{2^n}$ and $\sum_{a,\pi}X_{a,\pi}=\id_{2^n}$.
After receiving the message $\pi$ from Alice,
Bob performs his measurement described by $\set{Y^\pi_a}_a$,
where $Y^\pi_a\in\PSD_{3^n}$ and $\sum_aY^\pi_a=\id_{3^n}$ for each $\pi$.
Then their success probability is
$$\prob{\text{Success}} = \sum_{a, \pi}\Tr\left[\left(X_{a,\pi}\otimes Y^\pi_a\right)\cdot\left(\id\otimes\QEC^{\otimes n}\right)\left(\EPRN\right)\right].$$
%For $\ep\in(0,1)$, recall
%$\Pi_\epsilon\defeq\begin{bmatrix}1-\epsilon&\\&\epsilon\end{bmatrix}\otimes\id_2$
%and $\QECR(\rho)\defeq\ketbra{0}\otimes\rho+\tr\left(\rho\right)\cdot\ketbra{1}\otimes\id_2$,
%and we see that
%\begin{equation}\label{Equation-QEC-to-matrix-mult}
%  \QEC(\rho) = \Pi_\epsilon\cdot\QECR(\rho)
%\end{equation}
Our first lemma will be that Bob's strategy,
i.e. the measurement operators $Y^\pi_a$, will be of the form of \cref{app:def:goodform}.
That is,
$$Y^\pi_a=\sum_{S\subseteq[n]}Y^\pi_{a, S^c}\otimes\Err^{\otimes S}$$
where for all $S\subseteq[n]$, $Y^\pi_{a, S^c}$ lies in the subspace $\Matrix_{2^{[n]-S}}$.
As for an intuitive perspective,
this means Bob is the one who knows which qubits are polluted and which are not,
and can make his decision only based on the clean qubits.
Define
\[\overline{\Pi}_0=\begin{pmatrix}1 &&\\&1&\\&&0\end{pmatrix}~\text{and}~
\overline{\Pi}_1=\begin{pmatrix}0 &&\\&0&\\&&1\end{pmatrix},\]
and $\overline{\Pi}_b=\overline{\Pi}_{b_1}\otimes\ldots\otimes\overline{\Pi}_{b_1}$
for $b\in\set{0,1}^n$.
\begin{lemma}\label{lemma-optimal-Q-good-form-d}
  Let $X\in\PSD_{3^n}$ be of \goodform.
  Then $$\max_{Y\in\Matrix_{3^n}}\abs{\Tr \lrb{XY}}$$
  is achieved by a $Y$ of \goodform.
\end{lemma}
\begin{proof}
Notice that $X=\sum_{b\in\set{0,1}^n}\overline{\Pi}_bX\overline{\Pi}_b$. Suppose we have an optimal $Y\in\PSD_{3^n}$.
Let
$$\tilde{Y}\defeq\sum_{b\in\set{0,1}^n}\overline{\Pi}_b Y \overline{\Pi}_b.$$
Then
\begin{align*}
  \Tr \lrb{X\tilde{Y}} &= \sum_{b}\Tr\left[X\overline{\Pi}_bY\overline{\Pi}_b\right] \\
                 &= \Tr\left[\left(\sum_{b}\overline{\Pi}_bX\overline{\Pi}_b\right)Y\right] \\
                 &= \Tr\left[XY\right]. \\
\end{align*}
\end{proof}
\begin{cor}
  Bob's strategy can be replaced by measurement operators $\tilde{Y}^\pi_a$ of \goodform.
\end{cor}
\begin{proof}
Suppose Bob has an optimal strategy $\set{Y^\pi_a}$. Then
\begin{align*}
  \prob{\text{Success}} &= \sum_{a, \pi}\Tr\left[\left(X_{a,\pi}\otimes Y^\pi_a\right)\cdot\left(\id\otimes\QEC^{\otimes n}\right)\left(\EPRN\right)\right] \\
     &=\sum_{a,\pi}\Tr\left[\left(X_{a,\pi}\otimes{\QEC^\dagger}^{\otimes n}(Y^\pi_a)\right)\cdot\EPRN\right] \\
     &=2^{-n}\sum_{a,\pi}\Tr\left[X_{a,\pi}\cdot{\QEC^\dagger}^{\otimes n}(Y^\pi_a)^T\right] \\
     &=2^{-n}\sum_{a,\pi}\Tr\left[\QEC^{\otimes n}(X_{a,\pi})\cdot \br{Y^\pi_a}^T\right]. \\
\end{align*}
Notice that for each $a, \pi$, the operator $\QECN(X_{a, \pi})$ is of \goodform.
Let $$\tilde{Y}^\pi_a\defeq\sum_{b\in\set{0,1}^n}\overline{\Pi}_b Y^\pi_a\overline{\Pi}_b.$$
By \cref{lemma-optimal-Q-good-form-d},
we only need to check they form a POVM measurement. For each $\pi$,
$$\sum_a\tilde{Y}^\pi_a = \sum_{a, b}\overline{\Pi}_b Y^\pi_a\overline{\Pi}_b=\sum_b\overline{\Pi}_b\left(\sum_a Y^\pi_a\right)\overline{\Pi}_b=\sum_b\overline{\Pi}_b\id\overline{\Pi}_b=\sum_b\overline{\Pi}_b=\left(\overline{\Pi}_0+\overline{\Pi}_1\right)^{\otimes n}=\id_{3^n}.$$
\end{proof}
We are now ready to prove our main theorem for this section.
\begin{theorem}\label{thm:CRG}
  Let $0\le\epsilon\le1/2$.
  Suppose Alice and Bob share infinitely many pairs of $(\id\otimes\QEC)(\Phi)$,
  and Alice is allowed to send classical messages to Bob.
  Then for any $k\ge 1, \gamma\in(0, 1)$, to produce a common random string $R\in\set{0,1}^*$
  with $\Hmin{R} \ge k$ and a success probability at least $2^{-\gamma k}$,
  Alice needs to send at least
  $$\left(\ep^\prime(1-\gamma)-2\sqrt{\ep^\prime(1-\ep^\prime)\gamma}\right)\cdot k$$
  bits, where $\ep^\prime=\ep/2$.% and $c$ is the constant from \cref{hyper:thm:hc}.
\end{theorem}

\begin{remark}
	Grusuwami and Radhakrishnan~\cite{guruswami_et_al:LIPIcs:2016:5845} showed that it suffices for Alice to send $(\ep(1-\gamma)-2\sqrt{\ep(1-\ep)\gamma})k$ bits to Bob to achieve the same task, even if they share erased-noisy random strings. Hence, the lower bound given in \cref{thm:CRG} is tight up to a constant.
\end{remark}
\begin{proof}[Proof of \cref{thm:CRG}]
Our proof follows closely the proof for the classical counterpart~\cite{guruswami_et_al:LIPIcs:2016:5845}.
\begin{align*}
  \prob{\text{Success}} &= \sum_{a, \pi}\Tr\left[\left(X_{a,\pi}\otimes Y^\pi_a\right)\cdot\left(\id\otimes\QEC^{\otimes n}\right)\left(\EPRN\right)\right] \\
     &=\sum_{a,\pi}\Tr\left[\left(X_{a,\pi}\otimes{\QEC^\dagger}^{\otimes n}(Y^\pi_a)\right)\cdot\EPRN\right] \\
     &=2^{-n}\sum_{a,\pi}\Tr\left[X_{a,\pi}\cdot{\QEC^\dagger}^{\otimes n}(Y^\pi_a)^T\right] \\
     &=2^{-n}\sum_{a,\pi}\Tr\left[\QEC^{\otimes n}(X_{a,\pi})\cdot \br{Y^\pi_a}^T\right] \\
     &=2^{-n}\sum_{a,\pi}\Tr\left[\PiepN\QECR^{\otimes n}(X_{a,\pi})\cdot \br{Y^\pi_a}^T\right]. \\
\end{align*}
Then, by \cref{preliminary-Holder-3} with parameters $q$ and $q^*$, since $\br{Y^\pi_a}^T$ commutes with $\Pi_\ep$, we have
$$\prob{\text{Success}}\le2^{-n}\sum_\pi\left(\sum_a\Tr\left[\PiepN\QECR^{\otimes n}(X_{a,\pi})^q\right]\right)^{1/q}\cdot\left(\sum_a\Tr\left[\PiepN {\br{Y^\pi_a}^T}^{q^*}\right]\right)^{1/q^*}.$$
We first upper bound the second term. Since $\br{Y^\pi_a}^T\le\id_3$, we have ${\br{Y^\pi_a}^T}^{q^*}\le \br{Y^\pi_a}^T$, and thus
\begin{equation}\label{Equation-QEC-A1}
\sum_a\Tr\left[\Pi_\ep {\br{Y^\pi_a}^T}^{q^*}\right]\le\sum_a\Tr\left[\,\Pi_\ep \br{Y^\pi_a}^T\right]=\Tr\left[\Pi_\ep\sum_a\br{Y^\pi_a}^T\right]=\Tr\left[\Pi_\ep^{\otimes n}\id_{3^n}\right]=2^n.
\end{equation}
Applying \cref{Equation-QEC-A1}, we get
$$\prob{\text{Success}}\le2^{-n+n/q^*}\sum_\pi\left(\sum_a\Tr\left[\PiepN\QECR^{\otimes n}(X_{a,\pi})^q\right]\right)^{1/q}.$$
Now let $t$ be the number of bits sent from Alice to Bob.
Then the number of different choices of the message $\pi$ is at most $2^t$.
So by the concavity of the function $x\mapsto x^{1/q}$,
we can put the summation over $\pi$ into the parenthesis and get
$$\prob{\text{Success}}\le2^{-n+n/q^*+t/q^*}\left(\sum_{a,\pi}\Tr\left[\PiepN\QECRN(X_{a,\pi})^q\right]\right)^{1/q}.$$
We now apply hypercontractivity \cref{hyper:thm:hc} with $p=1+(1-\epsilon)^{1/c}(q-1)$ and $c=2$ to get
$$\prob{\text{Success}}\le2^{t/q^*}\left(\sum_{a,\pi}\left(2^{-n}\Tr\left[X_{a,\pi}^p\right]\right)^{q/p}\right)^{1/q}.$$
From the assumption that $\Hmin{R} \ge k$, we have $2^{-n}\Tr\left[X_{a,\pi}\right]\le2^{-k}$.
Notice $\norm{X_{a,\pi}}\le1$. Hence,
\begin{align*}
\prob{\text{Success}}&\le2^{t/q^*}\left(\sum_{a,\pi}\left(2^{-n}\Tr\left[X_{a,\pi}\right]\right)^{q/p}\right)^{1/q} \\
  &=2^{t/q^*}\left(\sum_{a,\pi}\left(2^{-n}\Tr\left[X_{a,\pi}\right]\right)\cdot\left(2^{-n}\Tr\left[X_{a,\pi}\right]\right)^{(q-p)/p}\right)^{1/q}\\
  &\le2^{t/q^*}\left(\sum_{a,\pi}\left(2^{-n}\Tr\left[X_{a,\pi}\right]\right)\cdot2^{-k(q-p)/p}\right)^{1/q}\\
  &=2^{t/q^*-k\frac{q-p}{pq}}.\\
\end{align*}
Let $\delta = q-1$, then $q=1+\delta$ and $p=1+(1-\epsilon)^{1/c}\delta$.
From the assumption that $\prob{\text{Success}}\geq 2^{-\gamma k}$,
we get for every $\delta > 0$,
\begin{align*}
t&\ge-\frac{(1+\delta)\gamma k}{\delta}+\frac{1-(1-\epsilon)^{1/c}}{1+(1-\epsilon)^{1/c}\delta}k \\
&\ge-\frac{(1+\delta)\gamma k}{\delta}+\frac{\epsilon/c}{1+(1-\epsilon/c)\delta}k\\
&=\left(\frac{\epsilon/c}{1+(1-\epsilon/c)\delta}-\frac{\gamma}{\delta}-\gamma\right)\cdot k.
\end{align*}
Maximizing over $\delta$, with $c=2$, we get
$$t\ge\p{(\epsilon/c)(1-\gamma)-2\sqrt{(\epsilon/c)(1-\epsilon/c)\gamma}}k.$$
\end{proof}
 % @dyj

%\section{Conclusion and Future Work}\label{section:Conclusion}
%\input{conclusion} % @buji

%\section{Acknowledgement}

\bibliographystyle{alpha}
\bibliography{main}

\appendix
\section{Proof of Lemma~\ref{hyper:lemma:psd}}\label{app:psd}

\hyperlemmapsd*

\begin{proof}
%    Our proof is divided into two steps. In the first step, we prove that there is no loss for us to restrict $X$ to be Hermitian, where the argument is similar to the one in~\cite{Wat05notes}. In the second step, we argue the  the value is unchanged even if we take the supremum over all PSD matrices. Our argument relies on a delicate analysis on the relation between eigenvalues and diagonal entries.
%
%    \paragraph{First part: restrict $X$ in Hermitian matrices}
%    Our goal is to show
%    $$
%	    \sup_{X\in \Matrix_{2^n},X\neq 0} \f{\normthree{X}_{\ep,q}}{\normthree{X}_p} \le \sup_{X\in \Herm_{2^n},X\neq 0} \f{\normthree{X}_{\ep,q}}{\normthree{X}_{p}}
%	$$
    Note that for any $X\in\Matrix_{2^n}$,
	$$
	    \p{2^{-n}\Tr \lrb{ \Pi_\ep \cdot \abs{D(X)}^q}}^{\f{1}{q}} = \p{\sum_{S\subseteq [n]}(1-\ep)^{n-\abs{S}}\ep^{\abs{S}}\normthree{\tr_S \lrb{X}}_q^q}^{\f{1}{q}}.
	$$
	
	Let $X\in \Matrix_{2^n}$ satisfy $\normthree{X}_p = 1$.
  For any $S\subseteq[n]$, let $\ps{Y}{S}\in \Matrix_{2^{n - \abs{S}}}$ satisfy $=\normthree{\ps{Y}{S}}_{q^*}=1$,
  where $q^*$ is the \Holder\ conjugate of $q$.
  Consider the singular value decompositions of $X$ and $\ps{Y}{S}$:
	$$
	    X=\sum_{i=1}^{2^n}l_i\ketbra{u_i}{v_i}\text{ and } \ps{Y}{S}=\sum_{j=1}^{m_{S}}r_{S,j}\ketbra{w_{S,j}}{x_{S,j}}
	$$
    where $m_S$ denotes the number of singular values of $\ps{Y}{S}$.
	
    Define $X_L,X_R,Y_{S^c,L},Y_{S^c,R}$ as follows:
	\begin{align*}
	    X_L=\sum_{i=1}^{2^n}l_i\ketbra{u_i}{u_i}\text{ and }\forall S\subseteq[n],Y_{S^c,L}=\sum_{j=1}^{m_{S}}r_{S,j}\ketbra{w_{S,j}}{w_{S,j}}, \\
	    X_R=\sum_{i=1}^{2^n}l_i\ketbra{v_i}\text{ and }\forall S\subseteq[n],Y_{S^c,R}=\sum_{j=1}^{m_{S}}r_{S,j}\ketbra{x_{S,j}}.
	\end{align*}
    We point out that $\normthree{X_L}_p=\normthree{X_R}_p=\normthree{Y_{S^c,L}}_{q^*}=\normthree{Y_{S^c,R}}_{q^*}=1$.
    Also, $X_L, X_R, Y_{S^c,L}, Y_{S^c,R}$ are positive semi-definite.
	
	Therefore,
	\begin{align*}
		&\wsum{n}{\abs{\tr\lrb{ Y_{S^c}^\dagger \cdot \p{\tr_S X}}}} = \wsum{n}\f{1}{2^{n-\abs{S}}}{\abs{\sum_{j=1}^{m_{S}} r_{S,j} \bra{w_{S,j}} \tr_S X \ket{x_{S,j}}}} \\
		& = \wsum{n}\f{1}{2^{n}}{\abs{\sum_{j=1}^{m_{S}}\sum_{k=1}^{2^\abs{S}} r_{S,j} \bra{w_{S,j},k} X \ket{x_{S,j},k}}} \\
		&= \wsum{n}\f{1}{2^{n}}{\abs{\sum_{i=1}^{2^n}\sum_{j=1}^{m_{S}}\sum_{k=1}^{2^\abs{S}} l_i r_{S^c,j} \bra{w_{S,j},k}\ket{u_i}\bra{v_i} \ket{x_{S,j},k}}} \\
		&\le \f{1}{2^n}\sqrt{\wsum{n}{\sum_{i,j,k} l_i r_{S,j} \abs{\bra{w_{S,j},k}\ket{u_i}}^2}}\cdot \sqrt{\wsum{n}{\sum_{i,j,k} l_i r_{S,j} \abs{\bra{v_i} \ket{x_{S,j},k}}^2}} \\
		&=\sqrt{\wsum{n}{\abs{\tr\lrb{ Y_{S^c,L}^\dagger  \cdot \tr_S X_L}}}} \cdot \sqrt{\wsum{n}{\abs{\tr\lrb{ Y_{S^c,R}^\dagger \cdot \tr_S X_R}}}} \\
		&\le \sqrt{\wsum{n}{\normthree{Y_{S^c,L}}_{q^*}\cdot \normthree{\tr_S X_L}_{q}}}\cdot \sqrt{\wsum{n}{\normthree{Y_{S^c,R}}_{q^*}\cdot \normthree{\tr_S X_R}_{q}}} \\
		&\le \sqrt{\p{\wsum{n}{\normthree{Y_{S^c,L}}_{q^*}^{q^*}}}^{\f{1}{q^*}} \p{\wsum{n}{\normthree{Y_{S^c,R}}_{q^*}^{q^*}}}^{\f{1}{q^*}} \normthree{X_L}_{\ep,q}  \normthree{X_R}_{\ep,q}} \\
		&= \sqrt{\normthree{X_L}_{\ep,q} \cdot \normthree{X_R}_{\ep,q}}.
	\end{align*}
	Taking the supremum over all $Y_{S^c}$, we conclude that for $\forall X\in \Matrix_{2^n}$, satisfying $\normthree{X}_p=1$,
	$$
	    \normthree{X}_{\ep,q}\le \sqrt{\normthree{X_L}_{\ep,q} \cdot \normthree{X_R}_{\ep,q}}.
	$$
  Without loss of generality, suppose $\normthree{X_L}_{\ep,q} \le \normthree{X_R}_{\ep,q}$, then
	$$
	    \normthree{X}_{\ep,q}\le \normthree{X_R}_{\ep,q}.
	$$
  Note that $X_R$ is positive semi-definite and $\normthree{X_R}_p=1$, so the lemma is proved.
\end{proof}

\section{Lemmas for Proof of Claim~\ref{claim:eq}}\label{app:induction}

We need two technical lemmas. Recall the two-point relative entropy function
$$\Ent_\ep\lrb{a,b} \defeq (1-\ep)\cdot a\ln a+\ep\cdot  b\ln b-\lrb{(1-\ep)a+\ep b}\ln\lrb{(1-\ep)a+\ep b}.$$

\begin{restatable}{lemma}{hypertechlemma}\label{app:lemma:techlemma}
    For $16b \ge a \geq b > 0$ and $0\le\ep\le1$,
    \begin{equation}\label{hyper:app:eq141}
         \Enteps\lrb{a,b} \leq 2\left(1-\ep\right)\left(a - b\right).
    \end{equation}
\end{restatable}

\begin{proof}
  Fix $a$ and $b$. Let
  $f(\ep) = \Enteps\lrb{a, b} - 2(1-\ep)(a-b)$.
  We see that $f(1) = 0$, so we only need to prove that $f$ is non-decreasing for $\ep$ in the interval $[0, 1]$.
  Calculating the derivatives of $f$ over $\ep$ we get
  $$\frac{d}{d\ep}f(\ep) = -a\ln a + b\ln b + 3(a-b) + (a-b)\ln\p{(1-\ep)a+\ep b},$$
  and
  $$\frac{d^2}{d\ep^2}f(\ep) = -\frac{(a-b)^2}{(1-\ep)a+\ep b}\le 0.$$
  Thus $f$ is concave. To prove $\frac{d}{d\ep}f(\ep)\ge 0$ for $\ep\in[0,1]$, we only need $\frac{d}{d\ep}f(1)\ge 0$, which is
  $$-a\ln a + 3(a-b) + a\ln b\ge 0$$
  for any $16b\ge a\ge b$.

  Now we only fix on $a$ and define the function $g: [a/16, a]\to\mathbb{R}$ as
  $$g(b) = -a\ln a + 3(a-b) + a\ln b.$$
  Again we see that $g(b)$ is concave because
  $\frac{d^2}{db^2}g(b) = -\frac{1}{b^2}\le 0$.
  To prove $g(b)\ge 0$ in the range $[a/16, a]$, it suffices to check
  \begin{itemize}
    \item $g(a) = -a\ln a+3(a-a)+a\ln a = 0$.
    \item $g(a/16) = \frac{45}{16}a - a\ln 16 = a\cdot \p{\frac{45}{16}-\ln 16}\ge 0$ since $\ln 16\approx2.77< \frac{45}{16} = 2.8125$.
  \end{itemize}
\end{proof}

\begin{lemma}\label{app:lemma:partialtracenormd}
  Let $A\in\Matrix_2\times\Matrix_n$ be positive semi-definite.
  Then for $q\ge 1$,
  \begin{equation}\label{app:eq:partialtracenormd}
    \tr\p{\tr_1\lrb{A}}^q\le\tr A^q\le 4^q\tr\p{\tr_1\lrb{A}}^q.
  \end{equation}
\end{lemma}
\begin{proof}
  The inequality $\tr\p{\tr_1\lrb{A}}^q\le\tr A^q$ is Proposition 1 in \cite{rastegin2012relations} in disguise.
  The inequality $\tr A^q\le 4^q\tr\p{\tr_1\lrb{A}}^q$ is a simple corollary of the inequality by Zhang~\cite{ZHANG2019258},
  who as Proposition 3.4 proved that
  $$4\tr_1\lrb{A}\otimes\id_2 = 2\Tr_1\lrb{A}\otimes \id_2\succeq A.$$
  Then
  $$4^q\p{\tr_1\lrb{A}}^q\otimes\id_2\succeq A^q$$
  and our lemma follows.
\end{proof}

\section{QEC hypercontractivity implies quantum depolarizing channel hypercontractivity}\label{appendix-proof-to-depolarizing-channel}
\begin{definition}[quantum depolarizing channel]
  Let $0\le\rho\le1$, the quantum depolarizing channel $\Delta_\rho$ is defined as
  \begin{equation}
  \Delta_\rho(X)\defeq\rho X+(1-\rho)\tr X\cdot \id_2.
  \end{equation}
\end{definition}
\begin{lemma}
  For any $p$ and $q$, suppose \cref{equation-qec-hc-introduction-d} holds for any matrix $X$.
  Then for any $Y$ of \goodform, it holds that
  \begin{equation}\label{appendix-equation-hiqx1}
    \abs{2^{-n}\Tr \lrb{Y^\dagger\QECN(X)}}\le\normthree{X}_p\p{2^{-n}\Tr\lrb{\PiepN \abs{Y}^{q^*}}}^{1/q^*},
  \end{equation}
  \begin{equation}\label{appendix-equation-hiqx2}
    \p{2^{-n}\Tr\abs{{\QECN}^\dagger(Y)}^{p^*}}^{1/p^*}\le \p{2^{-n}\Tr\lrb{\PiepN \abs{Y}^{q^*}}}^{1/q^*},
  \end{equation}
  where $p^*$ and $q^*$ are the \Holder\ conjugates of $p$ and $q$, respectively.
  Moreover, \cref{equation-qec-hc-introduction-d}, \cref{appendix-equation-hiqx1}, \cref{appendix-equation-hiqx2} are all equivalent to each other.
\end{lemma}
\begin{proof}
  \begin{itemize}
  \item \cref{equation-qec-hc-introduction-d} $\implies$ \cref{appendix-equation-hiqx1}
  \begin{align*}
    \abs{2^{-n}\Tr \lrb{Y^\dagger\QECN(X)}} &= \abs{2^{-n}\Tr\lrb{\PiepN Y^\dagger\QECRN(X)}}  \\
    &\le \p{2^{-n}\Tr\lrb{\PiepN\abs{\QECRN(X)}^q}}^{1/q}\p{2^{-n}\Tr\lrb{\PiepN \abs{Y}^{q^*}}}^{1/q^*}\\
    &\le \normthree{X}_p\p{2^{-n}\Tr\lrb{\PiepN \abs{Y}^{q^*}}}^{1/q^*}.
  \end{align*}
  \item \cref{appendix-equation-hiqx1} $\implies$ \cref{equation-qec-hc-introduction-d}
  \begin{align*}
  &\left(2^{-n}\Tr\left[\PiepN\abs{\QECRN(X)}^q\right]\right)^{1/q}\\
    =& \max_{Y}\set{\abs{2^{-n}\Tr \lrb{Y^\dagger \p{\PiepN}^{1/q}\QECRN(X)}}: Y\in\Matrix_{3^n}, \normthree{Y}_{q^*}\le1} \\
    =& \max_{Y}\set{\abs{2^{-n}\Tr \lrb{Y^\dagger \PiepN\QECRN(X)}}: Y\in\Matrix_{3^n}, \p{2^{-n}\Tr\lrb{\PiepN \abs{Y}^{q^*}}}^{1/q^*}\le 1} \\
    \le&\normthree{X}_p.
  \end{align*}
  where the maximizations are over $Y$ of \goodform\ and
  the first equality follows from \cref{preliminary-Holder-2} and \cref{lemma-optimal-Q-good-form-d}.
  \item \cref{appendix-equation-hiqx2} $\implies$ \cref{appendix-equation-hiqx1}
  \begin{align*}
    \abs{2^{-n}\Tr \lrb{Y^\dagger\QECN(X)}} &= \abs{2^{-n}\Tr \lrb{{\QECN}^\dagger(Y^\dagger)X }} \\
    &\le \normthree{X}_p\p{2^{-n}\Tr\abs{{\QECN}^\dagger(Y)}^{p^*}}^{1/p^*}\\
    &\le \normthree{X}_p\p{2^{-n}\Tr\lrb{\PiepN \abs{Y}^{q^*}}}^{1/q^*}.
  \end{align*}
  \item \cref{appendix-equation-hiqx1} $\implies$ \cref{appendix-equation-hiqx2}
  \begin{align*}
    \p{2^{-n}\Tr\abs{{\QECN}^\dagger(Y)}^{p^*}}^{1/p^*} &= \max_X\set{\abs{2^{-n}\Tr \lrb{X^\dagger{\QECN}^\dagger(Y)}}: \normthree{X}_p\le 1} \\
      &= \max_X\set{\abs{2^{-n}\Tr \lrb{\QECN(X^\dagger) Y}}: \normthree{X}_p\le 1} \\
      &\le \normthree{X}_p\p{2^{-n}\Tr\lrb{\PiepN \abs{Y}^{q^*}}}^{1/q^*} \\
      &\le \p{2^{-n}\Tr\lrb{\PiepN \abs{Y}^{q^*}}}^{1/q^*}.
  \end{align*}
  where the first equality is \cref{preliminary-Holder-2}.
  \end{itemize}
\end{proof}

\begin{lemma}\label{appendix-depolarizing-final}
  Suppose there exist parameters $\ep, p, q, r$ such that \cref{equation-qec-hc-introduction-d} holds for
  \begin{itemize}
    \item $q\gets r^*$ and $p\gets q^*$ and
    \item $q\gets r$ and $p\gets p$,
  \end{itemize}
  where $r^*$ and $q^*$ are the \Holder\ conjugates of $r$ and $q$, respectively.
  Then for any matrix $X$, 
  for $\rho\leq1-\ep$, we have the hypercontractivity inequality
  \begin{equation}\label{appendix-depolarizing-channel-dda}
    \normthree{\Delta_\rho^{\otimes n}(X)}_q\le\normthree{X}_p.
  \end{equation}
\end{lemma}
\begin{proof}
  Note that
  $$\Delta_\rho^{\otimes n}(X) = {\QECN}^\dagger\p{\QECRN(X)}.$$
  So we have
  \begin{align*}
    \normthree{\Delta_\rho^{\otimes n}(X)}_q &= \p{2^{-n}\Tr\abs{{\QECN}^\dagger(\QECRN(X))}^{q}}^{1/q}\\
      &\le\p{2^{-n}\Tr\lrb{\PiepN \abs{\QECRN(X)}^{r}}}^{1/r} \\
      &\le \normthree{X}_p.
  \end{align*}
\end{proof}
\begin{cor}
  Suppose \cref{equation-qec-hc-introduction-d} holds for all $1-\ep\le\left(\frac{p-1}{q-1}\right)^c$ and any matrix $X$, then
  \cref{appendix-depolarizing-channel-dda} holds for all $\rho=1-\ep\le\left(\frac{p-1}{q-1}\right)^{c/2}$ and any matrix $X$.
\end{cor}
\begin{proof}
  Let $r=1+(p-1)/(1-\ep)^{1/c}$, then we can check that
  $$\p{\frac{r-1}{q-1}}^c = 1-\ep,$$
  and
  $$\p{\frac{p-1}{r-1}}^c = 1-\ep.$$
  Then we can apply \cref{appendix-depolarizing-final} with the parameters $\ep, p, q, r$ to prove \cref{appendix-depolarizing-channel-dda}.
\end{proof}

\section{Proof of Hypercontractivity for \texorpdfstring{$1\le p\le 2\le q$}{1 ≤ p ≤ 2 ≤ q}}\label{appendix-proof-HC-1p2q}
\begin{lemma}[\cite{2003CMaPh.242..531K}]\label{appendix-89fh89aedf}
  Let $M$ be a $2n\times2n$ positive semidefinite matrix.
  It can be written in the block form
  \begin{equation}
    M=\begin{pmatrix}X&Y\\Y^\dagger&Z\end{pmatrix}
  \end{equation}
  Define the $2\times 2$ matrix
  \begin{equation}
    m_p=\begin{pmatrix}\normsub{X}{p}&\normsub{Y}{p}\\\normsub{Y}{p}&\normsub{Z}{p}\end{pmatrix}
  \end{equation}
  Then $m_p$ is positive semidefinite and
  \begin{itemize}
    \item for $1\le p\le 2$
      \begin{equation}
        \normsub{M}{p}\ge\normsub{m_p}{p}
      \end{equation}
    \item for $2\le p\le\infty$
      \begin{equation}
        \normsub{M}{p}\le\normsub{m_p}{p}
      \end{equation}
  \end{itemize}
\end{lemma}
\begin{lemma}\label{appedix-4ch7843hh3}
  Let $X, Y$ be two $2\times 2$ real valued positive semidefinite matrices such that for all $i, j\in\set{1,2}$
  \begin{equation}
    X_{ij} \le Y_{ij}
  \end{equation}
  then for any Schatten $p$-norm $\norm{\cdot}$ we have
  \begin{equation}
    \norm{X} \le \norm{Y}
  \end{equation}
\end{lemma}
\begin{proof}
  Define the intermediate matrix
  \begin{equation}
    Z=\begin{pmatrix}Y_{11} & X_{12}\\X_{12} & Y_{22}\end{pmatrix} = X + \begin{pmatrix}Y_{11}-X_{11}&\\&Y_{22}-X_{22}\end{pmatrix}
  \end{equation}
  it is clear that $X\preceq Z$ and thus $\norm{X}\le\norm{Z}$.
  We only need to prove $\norm{Z}\le\norm{Y}$.

  Let $\lambda_1\ge\lambda_2\ge0$ be the eigenvalues of $Y$
  and $\gamma_1\ge\gamma_2\ge0$ be the eigenvalues of $Z$.
  Since $\Tr Y = \Tr Z$, we have $\lambda_1+\lambda_2 = \gamma_1+\gamma_2$.
  Since $\det(Y)\le\det(Z)$, we have $\lambda_1\cdot\lambda_2 \le \gamma_1\cdot\gamma_2$.
  Thus $\set{\lambda_1, \lambda_2}$ majorises $\set{\gamma_1, \gamma_2}$ and thus $\norm{Y}\ge\norm{Z}$.
\end{proof}
\begin{lemma}\label{Lemma-QEC-HC-Induction-2}
  Let $A\in\Matrix_{2^n}$, then for all $1\le p\le 2\le q$ and $0\le\ep\le 1$ satisfying $1-\ep\le\frac{p-1}{q-1}$,
  \begin{equation}
    \left(\Tr \lrb{\PiepN\abs{\QECRN(X)}^q}\right)^{1/q}\le\normsub{X}{p}\cdot 2^{n\cdot\p{\frac{1}{q}-\frac{1}{p}}}.
  \end{equation}
\end{lemma}
\begin{proof}
  %Since $\Omega\otimes\QEC$ is completely positive,
  %it's $p\to q$ norm is achieved on a positive semidefinite matrix.
  %So we only need to prove that for all $A\in\PSD$,
  %$$\normsub{(\Omega\otimes\QEC)(A)}{q} \le \normsub{\Omega}{p\to q}\normsub{\QEC}{p\to q}\normsub{A}{p}$$
  By \cref{hyper:lemma:psd}, we can assume $X$ to be positive semi-definite.
  For $n=1$, the inequality collapses to the classical case and has been proven by Nair and Wang~\cite{7541363}.
  For $n>1$, we can write $X$ as a block matrix
  $$X = \begin{pmatrix}X_{11} & X_{12}\\ X_{21} & X_{22}\end{pmatrix}$$
  then
  $$\left(\PiepN\right)^{1/q}\QECRN(X) = \begin{pmatrix}Z_{11} & Z_{12} & \\Z_{21} & Z_{22} & \\  & & Z_{33}\end{pmatrix}$$
  %  \begin{pmatrix}(1-\ep)\Piep^{\otimes n-1}\QECR^{\otimes n-1}(A_{11}) & (1-\ep)\Piep^{\otimes n-1}\QECR^{\otimes n-1}(A_{12}) & (1-\ep)\Omega(A_{22}) & 0 \\ 0 & 0 & \ep(\Omega(A_{11})+\Omega(A_{22}))\end{pmatrix}
  where we let $Y_{ij} = \left(\Piep^{\otimes n-1}\right)^{1/q}\QECR^{\otimes n-1}(X_{ij})$ for $i,j\in\set{1,2}$ and
  $Z_{ij}=(1-\ep)^{1/q}Y_{ij}$ for $i,j\in\set{1,2}$ and $Z_{33}=(2\ep)^{1/q}\f{Y_{11}+Y_{22}}{2}$.
  Since $q\ge2$, \cref{appendix-89fh89aedf} implies
  \begin{align*}
  \left(\Tr\lrb{\PiepN\QECRN(X)^q}\right)^{1/q} &= \p{\Tr \lrb{\begin{pmatrix}Z_{11} & Z_{12} & \\Z_{21} & Z_{22} & \\  & & Z_{33}\end{pmatrix}^q}}^{1/q} \\
      &=\normsub{\begin{pmatrix}
          \normsub{\begin{pmatrix}
              Z_{11} & Z_{12} \\ Z_{21} & Z_{22}
          \end{pmatrix}}{q} &
          \\
             & \normsub{Z_{33}}{q}
          \end{pmatrix}}{q}\\
      &\le\normsub{\begin{pmatrix}\normsub{\begin{pmatrix}\normsub{Z_{11}}{q} & \normsub{Z_{12}}{q} \\ \normsub{Z_{12}}{q} & \normsub{Z_{22}}{q}\end{pmatrix}}{q} &\\ & \normsub{Z_{33}}{q}\end{pmatrix}  }{q}  \\
  \end{align*}
  By triangle inequality we have
  $$\normsub{Z_{33}}{q} = (2\ep)^{1/q}\normsub{\f{Y_{11}+Y_{22}}{2}}{q}\le\f{(2\ep)^{1/q}}{2}\normsub{Y_{11}}{q} + \f{(2\ep)^{1/q}}{2}\normsub{Y_{22}}{q}$$
  Thus
  \begin{equation}
\left(\Tr\lrb{\PiepN\QECRN(X)^q}\right)^{1/q}\le\normsub{\Piep^{1/q}\QECR\begin{pmatrix}\normsub{Y_{11}}{q}&\normsub{Y_{12}}{q}\\\normsub{Y_{12}}{q}&\normsub{Y_{22}}{q}\end{pmatrix}}{q}
  \end{equation}
  By induction hypothesis we also have
  \begin{equation}
    \normsub{Y_{ij}}{q} = \left(\Tr\Piep^{\otimes n-1}\QECR^{\otimes n-1}(X_{ij})^{q}\right)^{1/q}\le\normsub{X_{ij}}{p}\cdot 2^{(n-1)\cdot\p{\frac{1}{q}-\frac{1}{p}}}
  \end{equation}
  thus by \cref{appedix-4ch7843hh3} and again by induction hypothesis
  \begin{multline}
\left(\Tr\lrb{\PiepN\QECRN(X)^q}\right)^{1/q}\le\normsub{\Piep^{1/q}\QECR\begin{pmatrix}\normsub{X_{11}}{p}&\normsub{X_{12}}{p}\\\normsub{X_{12}}{p}&\normsub{X_{22}}{p}\end{pmatrix}}{q}\cdot 2^{(n-1)\cdot\p{\frac{1}{q}-\frac{1}{p}}}\\
\le\normsub{\begin{pmatrix}\normsub{X_{11}}{p}&\normsub{X_{12}}{p}\\\normsub{X_{12}}{p}&\normsub{X_{22}}{p}\end{pmatrix}}{p}\cdot 2^{n\cdot\p{\frac{1}{q}-\frac{1}{p}}}
  \end{multline}
  Finally we use \cref{appendix-89fh89aedf} again for $1\le p\le 2$ which implies
  $$\normsub{\begin{pmatrix}\normsub{X_{11}}{p} & \normsub{X_{12}}{p}\\\normsub{X_{12}}{p} & \normsub{X_{22}}{p}\end{pmatrix}}{p}\le\normsub{\begin{pmatrix}X_{11} & X_{12}\\X_{12}^\dagger & X_{22}\end{pmatrix}}{p} = \normsub{X}{p}$$
  and the lemma is proved.
\end{proof}

\end{document}